\pdfoutput=1
\documentclass[11pt,preprint]{elsarticle}

\usepackage[margin=1in]{geometry}
\usepackage{amsmath,amssymb,amsthm,mathtools,bm}
\usepackage[T1]{fontenc}
\usepackage{lmodern}
\usepackage{microtype}
\usepackage{graphicx}
\usepackage{subcaption}
\usepackage{booktabs}
\usepackage{float}
\usepackage{placeins}
\usepackage[normalem]{ulem}
\usepackage{algorithm}
\usepackage{algpseudocode}
\usepackage{xcolor}
\usepackage{hyperref}
\hypersetup{colorlinks=true,linkcolor=blue,citecolor=blue,urlcolor=blue}
\allowdisplaybreaks

\newtheorem{proposition}{Proposition}
\newtheorem{lemma}{Lemma}

\newtheorem{remark}{Remark}

\newcommand{\Kc}{\mathcal{K}}
\newcommand{\pinf}{\pi_{\infty}}
\newcommand{\pinfj}[1]{\pi_{\infty,#1}}

\algnewcommand{\Phase}[2]{\Statex\textit{#1}\hfill{\itshape\footnotesize #2}}
\algnewcommand{\PhaseTwoLine}[2]{\Statex\textit{#1}\Statex\hfill{\itshape\footnotesize #2}}

\journal{Journal of Computational Physics}

\begin{document}

\begin{frontmatter}
\title{Anchored HLLD: A Riemann solver for a hypoelastic
diffuse-interface model of compressible fluid--solid interaction}
\author[caltech]{Yanjun Zhang\corref{cor1}}
\ead{yzhang6@caltech.edu}
\cortext[cor1]{Corresponding author}
\author[caltech]{Tim Colonius}
\ead{colonius@caltech.edu}
\affiliation[caltech]{organization={Department of Mechanical and Civil Engineering, California Institute of Technology},
city={Pasadena}, statesep={}, state={CA}, postcode={91125}, country={USA}}
\begin{abstract}
We present an HLLD Riemann solver for a multicomponent hypoelastic diffuse-interface model of compressible fluid--solid interaction, with no separately computed interface and a deviatoric stress advanced by an objective hypoelastic rate law outside conservation-law form.
The difficulty specific to this model is that the stress evolution is nonconservative and, unlike the familiar volume-fraction terms of Kapila-type models, cannot be left outside the Riemann solve once the shear waves are resolved.
The shear coefficient $\rho(G+\tau_{xx})$ sets the impedance of the resolved shear waves and so belongs to their eigenstructure.
\emph{Anchored HLLD} carries it inside the solve by evaluating the HLLD fan twice at each face, anchored in turn to each adjacent cell-centered state, with every cell advancing from the flux anchored to its own state.
HLL, HLLC, and anchored HLLD form a three-tier hierarchy of increasing wave resolution in one framework.
We prove local structural properties: consistency at the reconstruction order, exact species-mass conservation, and tangential-momentum and energy drift vanishing under refinement at the reconstruction's face-jump order on smooth fields, with one-dimensional tests illustrating these structural properties and demonstrating first-order drift convergence on discontinuous fields.
Across one-dimensional tests and two-dimensional production simulations in the open-source solver MFC, anchored HLLD resolves both the material contact and the shear waves, HLLC the contact alone, and HLL neither sharply.
Where refinement studies are shown, anchored HLLD converges toward exact or semi-analytic solutions.
A three-dimensional shock--ellipsoid interaction demonstrates the solver with all six stress components active.
\end{abstract}


\end{frontmatter}

\clearpage
\section{Introduction}

Compressible fluid--solid interaction arises in a wide range of applications, from underwater shock loading of structures to biomedical acoustics and geophysical wave propagation.
A natural numerical framework for such problems is a fixed-grid Eulerian method in which material interfaces are represented implicitly through evolved volume fractions, with no separately computed interface description.
This \emph{interface-capturing} approach is often called a diffuse-interface method, as the captured interface relaxes into an artificial mixture zone a few grid cells wide.
Interfaces captured this way are prone to spurious pressure oscillations under fully conservative schemes.
The foundational fluid formulations prevent them by transporting the material and equation-of-state quantities in nonconservative form rather than conservatively: the shock-capturing scheme of Shyue~\cite{Shyue1998}, the seven-equation model of Saurel and Abgrall~\cite{SaurelAbgrall1999}, and the reduced five-equation model of Kapila and coworkers~\cite{Kapila2001}.
The resulting systems remain hyperbolic, inheriting the mathematical structure of single-phase gas dynamics and thereby admitting the established toolkit of shock-capturing Riemann solvers.
Subsequent work~\cite{AlahyariBeigJohnsen2015} makes this precise, identifying the transport form that preserves interface pressure and temperature equilibrium.

Extension to solid mechanics has followed two constitutive routes.
The hyperelastic branch evolves a kinematic measure such as the deformation gradient, stretch, or Finger tensor and derives the stress from a stored elastic energy density.
A conservative single-material Eulerian formulation~\cite{GavrilyukFavrieSaurel2008} was extended to the multicomponent diffuse-interface setting by Favrie and coworkers~\cite{FavrieSaurelGavrilyuk2009,FavrieGavrilyuk2012,NdanouFavrieGavrilyuk2015}.
The branch now includes both fully diffuse multicomponent constructions~\cite{Barton2019} and sharp-interface treatments that couple the phases through level-set or ghost-fluid machinery~\cite{LopezOrtegaLombardiniPullinMeiron2014,WangAdamiAdams2024}, as well as high-order schemes outside the HLL family~\cite{WestAdlerLele2024}.

The second branch, which we adopt here, evolves the deviatoric stress directly through an objective rate law, with the pressure supplied by the equation of state.
This \emph{hypoelastic} approach fits naturally into the multicomponent interface-capturing framework, coupling to the existing volume-fraction and equation-of-state infrastructure without an additional evolved kinematic strain measure.
The augmented system remains hyperbolic, with shear waves joining the acoustic and contact families, so that Riemann-solver machinery applies directly.
The branch carries a known constitutive tradeoff.
Hypoelastic rate laws with constant moduli built on the classical objective rates, including the Lie derivative used here, are not exactly integrable to an elastic relation derived from a stored energy~\cite{SimoPister1984}, so the thermodynamic structure of the hyperelastic branch is not inherited.
For the small-deformation regime targeted here this is a standard and well-understood compromise.
The consequence that matters for this paper is different.
The stress evolution introduces nonconservative products, which create a specific numerical difficulty described below.

The model we extend is the five-equation Eulerian interface-capturing framework of Rodriguez and Johnsen~\cite{RodriguezJohnsen2019}, in which gases, liquids, and solids obey a common equation of state with spatially varying material parameters, and small-deformation elasticity is introduced through a Lie-derivative-based hypoelastic stress evolution.
The original treatment uses an HLL Riemann solver and is implemented in the open-source exascale solver MFC~\cite{MFC2021,MFC2026}, within which all results in this paper are produced.
Subsequent work by Rodriguez, Johnsen, and Powell~\cite{RodriguezJohnsenPowell2019} developed an AUSM\textsuperscript{+}-up flux-vector-splitting variant.
Neither resolves the shear waves as a distinct wave family, a gap we fill with an HLLD Riemann solver.

To explain the numerical difficulty, we recall the Riemann-solver hierarchy.
The HLL solver of Harten, Lax, and van Leer~\cite{HartenLaxVanLeer1983} uses two waves bounding the fastest signals and averages the solution between them.
HLLC~\cite{ToroSpruceSpeares1994} adds a third wave restoring the material contact discontinuity.
HLLD~\cite{MiyoshiKusano2005}, originally developed for ideal magnetohydrodynamics to resolve Alfv\'{e}n waves, adds an inner pair of waves on either side of the contact.
In the hypoelastic context these inner waves are shear waves, with speed $c_s = \sqrt{(G+\tau_{xx})/\rho}$.
HLL and HLLC smear the shear pair, averaging it across the outer wave fan.
Only HLLD resolves it as a distinct wave family.

The difficulty specific to the hypoelastic model is that the nonconservative stress coefficient $\rho(G+\tau_{xx})$ simultaneously enters a source term in the stress evolution equations \emph{and} sets the shear-wave impedance, making it part of the eigenstructure of the very waves HLLD resolves.
For familiar nonconservative terms such as the $\alpha\,\nabla\!\cdot\!u$ and $\Kc\,\nabla\!\cdot\!u$ terms in the volume-fraction equation of Kapila-type five-equation fluid models, with $\Kc$ the compressibility coefficient, the coefficient does not set the speed or impedance of any resolved wave in the Riemann solver, so it can be applied outside the Riemann solve as a source term multiplying a face-trace velocity.
This Kapila-type out-of-solver principle extends to HLLC for the hypoelastic system, since the shear coefficient does not enter the contact wave HLLC resolves.
But it does not extend to HLLD.
The shear-wave impedance $\sqrt{\rho(G+\tau_{xx})}$ is inseparable from the wave being resolved, so the nonconservative coefficient must enter the Riemann solve itself.

We accomplish this through a \emph{dual-pass} face evaluation: the HLLD solver is called twice at each face with the same reconstructed left and right states, but with the nonconservative coefficients frozen at the cell-centered state of each adjacent cell in turn, and each cell advances using the flux from the evaluation anchored to its own state.
We call the resulting solver \emph{anchored HLLD}.
We prove that sampled HLLD inner shear branches with active jumps do not admit the out-of-solver face representation (Proposition~\ref{prop:anchoring-criterion}).
In the fluid limit the construction reduces to a conventional single-pass fluid HLLC update (Section~\ref{subsec:reductions}).

The construction involves no interface localization and no ghost states.
Every cell is updated by the same dual-pass evaluation regardless of its volume-fraction content.
It therefore advances genuine fluid--solid mixture cells with intermediate volume fractions by the same operations as pure cells, a state that interface-localizing ghost-fluid couplings do not naturally represent.
A validated physical closure for such media, for example saturated porous solids, is future work.

Several existing solvers warrant explicit comparison.
In the hyperelastic branch, Barton~\cite{Barton2019} constructs a multicomponent diffuse-interface fluid--solid HLLD with the stretch tensor as part of the Riemann state.
It is the closest existing construction, distinguished by its constitutive branch rather than by interface treatment.
L\'{o}pez Ortega and coworkers~\cite{LopezOrtegaLombardiniPullinMeiron2014} and Wang, Adami, and Adams~\cite{WangAdamiAdams2024} use Eulerian level-set ghost-fluid HLLD constructions with evolved kinematic measures.
In the hypoelastic branch, the HLLD-e solver of Zhao and coworkers~\cite{ZhaoWangJiaWang2023} resolves elastic shear waves in two-dimensional hypoelasticity, and the modified-ghost-fluid method of Liu, Xie, and Khoo~\cite{LiuXieKhoo2008} couples fluids and solids, but both localize the material interface rather than using the diffuse representation here.
The HLLC solvers of Liu, Cheng, and Liu~\cite{LiuChengLiu2019} and Li, Zhai, and Shen~\cite{LiZhaiShen2022} are cell-centered Lagrangian rather than Eulerian.
The hypoelastic HLL-type solver of Zeng and coworkers~\cite{ZengEtAl2023} and the stress-continuity solver of Zeng, Cao, and coworkers~\cite{ZengEtAl2024} address single-material solids rather than fluid--solid or multicomponent configurations.
Two recent constructions are adjacent.
Serezhkin and Menshov~\cite{SerezhkinMenshov2025} develop a two-material Eulerian diffuse-interface hypoelastic model with an HLLC-type solver in a Baer--Nunziato relaxation framework, so the shear waves are not resolved as a distinct family.
Zeng and coworkers~\cite{ZengEtAl2026} treat the nonconservative products of hypoelasticity with a path-conservative scheme built on Hooke's-law-preserving paths, for a single-material solid rather than the multicomponent fluid--solid setting here.
General nonconservative frameworks include the Dal Maso--LeFloch--Murat theory~\cite{DalMasoLeFlochMurat1995}, path-conservative Godunov constructions~\cite{Pares2006,MunozRuizPares2007,CastroDiaz2013}, and HLLEM/HLLI-type solvers~\cite{DumbserBalsara2016}.
The solid and fluid--solid literature above instead develops purpose-built characteristic-wave HLLC and HLLD solvers, and we follow that line with anchored HLLD tailored to the multicomponent hypoelastic diffuse-interface setting.

The HLL, HLLC, and anchored HLLD members form a three-tier hierarchy of increasing wave resolution within one framework.
HLL is the two-wave baseline of Rodriguez and Johnsen.
HLLC restores the material contact wave through a new adaptation of the established Kapila-type out-of-solver principle, evaluated with the HLLC face trace.
Anchored HLLD, the principal new construction, additionally resolves the shear pair with the nonconservative coefficient carried inside the solve.
To our knowledge this is the first HLLD solver for a multicomponent hypoelastic diffuse-interface fluid--solid model with no separately computed interface description.
Two per-cell-frozen fluxes per face also appear in the double-flux method of Abgrall and Karni~\cite{AbgrallKarni2001}, which freezes a thermodynamic closure to keep material interfaces oscillation-free.
The similarity is structural only.
There the frozen quantity is an equation-of-state closure of a conservative flux, whereas here it is the coefficient of a nonconservative product that sets the impedance of a resolved inner wave.

We support the construction with analytical results on consistency, contact preservation, conservation and its drift rate, and the structural reductions linking the three tiers.
We verify the solvers on targeted one- and two-dimensional tests that isolate the inactive- and active-shear regimes.
With the shear active, anchored HLLD resolves both the contact and the shear waves, HLLC the contact alone, and HLL neither sharply.
At production scale, in acoustic and shock scattering by an epoxy cylinder, HLLC and anchored HLLD reduce the diffusion of the material interface compared to HLL, and anchored HLLD gives the sharpest stress response in both cases and lies closest to the semi-analytic reference in the acoustic case.
The anchored construction conserves the species masses exactly, which the nonconservative ghost-fluid methods do not~\cite{ZhaoWangJiaWang2023,LiuXieKhoo2008,XuLiu2011}, and its drift in the tangential momenta and total energy converges to zero at the reconstruction's face-jump order on smooth data (\ref{sec:analysis}).
One-dimensional contact and shock tests also demonstrate first-order drift convergence.

The remainder of the paper is organized as follows.
Section~\ref{sec:governing} presents the governing equations, the conservative--nonconservative split, and the characteristic wave structure.
Section~\ref{sec:method} develops the HLL, HLLC, and anchored HLLD constructions and their reductions.
Section~\ref{sec:results-1d} reports one-dimensional verification, including smooth eigenmode convergence, contact and Fresnel tests, and shear-resolution comparisons.
Section~\ref{sec:results-2d} assesses the solvers at production scale on transient acoustic scattering by an epoxy cylinder against a semi-analytic reference and on a two-dimensional shock--cylinder interaction, and demonstrates anchored HLLD on a three-dimensional shock--ellipsoid interaction.
Section~\ref{sec:conclusions} concludes.
\ref{sec:analysis} establishes the analytical properties of the out-of-solver update and the anchored HLLD face operator.
\ref{app:3d} extends the system and the anchored HLLD construction to three dimensions.
\ref{app:complex} presents a shock interaction with complex inclusions.

\FloatBarrier
\section{Governing equations and wave structure}\label{sec:governing}

\subsection{Governing equations}

We consider the inviscid two-dimensional multicomponent hypoelastic model in Cartesian coordinates,
\begin{equation}
\label{eq:full2d}
\begin{aligned}
\partial_t \alpha_1 + \nabla\!\cdot(\alpha_1 \bm{u}) &= (\alpha_1+\Kc)\,\nabla\!\cdot\bm{u},\\
\partial_t m_j + \nabla\!\cdot(m_j\bm{u}) &= 0, \qquad j=1,2,\\
\partial_t(\rho \bm{u}) + \nabla\!\cdot\bigl(\rho\,\bm{u}\otimes \bm{u} + p\,\bm{I} - \bm{\tau}\bigr) &= \bm{0},\\
\partial_t(\rho E) + \nabla\!\cdot\bigl((\rho E+p)\bm{u} - \bm{\tau}\,\bm{u}\bigr) &= 0,\\
\partial_t(\rho\tau_{ij}) + \nabla\!\cdot(\rho\tau_{ij}\bm{u}) &= \rho\Bigl(\tau_{kj}\,\partial_k u_i + \tau_{ik}\,\partial_k u_j - \tau_{ij}\,\partial_k u_k\Bigr)
+ 2\rho G\,d_{ij},
\end{aligned}
\end{equation}
with
\begin{equation}
\label{eq:modeldefs}
\bm{u}=(u,v)^T,
\qquad
m_1=\alpha_1\rho_1,
\qquad
m_2=\alpha_2\rho_2,
\qquad
\alpha_2=1-\alpha_1,
\qquad
\rho=m_1+m_2,
\end{equation}
and deviatoric strain rate
\begin{equation}
\label{eq:dij}
d_{ij}=\frac12\left(\partial_j u_i+\partial_i u_j\right)-\frac13(\nabla\!\cdot\bm{u})\,\delta_{ij}.
\end{equation}
In two dimensions the total energy is
\begin{equation}
\label{eq:Edef}
E=e+\frac12(u^2+v^2),
\end{equation}
where the internal energy carries a thermal and an elastic contribution,
\begin{equation}
\label{eq:caloric}
\rho e=\frac{p+\gamma\pinf}{\gamma-1}+e_{\mathrm{el}},
\qquad
e_{\mathrm{el}}=\frac{\tau_{xx}^2+\tau_{yy}^2+2\tau_{xy}^2}{4G},
\end{equation}
with \(\gamma\) and \(\pinf\) the stiffened-gas parameters and \(e_{\mathrm{el}}\) the in-plane quadratic elastic energy per unit volume, following Rodriguez and Johnsen~\cite{RodriguezJohnsen2019}.
The pressure is recovered from the conserved energy by inverting Eqs.~\eqref{eq:Edef}--\eqref{eq:caloric}, \(p=(\gamma-1)\bigl(\rho E-\tfrac12\rho(u^2+v^2)-e_{\mathrm{el}}\bigr)-\gamma\pinf\).
The mixture parameters obey the standard stiffened-gas mixture rules together with a volume-fraction-weighted shear modulus,
\begin{equation}
\label{eq:mixrules}
\frac{1}{\gamma-1}=\sum_{j}\frac{\alpha_j}{\gamma_j-1},
\qquad
\frac{\gamma\,\pinf}{\gamma-1}=\sum_{j}\alpha_j\,\frac{\gamma_j\pinfj{j}}{\gamma_j-1},
\qquad
G=\sum_{j}\alpha_j G_j,
\end{equation}
with \(j=1,2\) the phase index.
The acoustic closure is Wood's mixture relation over the phase bulk moduli~\cite{Wood1930},
\begin{equation}
\label{eq:eos}
a^2=\frac{1}{\rho}\left[\sum_{j}\frac{\alpha_j}{B_j}\right]^{-1},
\qquad
B_j=\gamma_j\bigl(p+\pinfj{j}\bigr),
\end{equation}
with \(B_j\) the bulk modulus of phase \(j\).
The symbol \(\Kc\) denotes the Kapila-type correction coefficient of the volume-fraction equation, distinguished from the side indices \(K\in\{L,R\}\) used below.
The mixture coefficients \(\Kc\) and \(G\) are evaluated from cell-centered primitive data in the numerical method.
The Kapila-type coefficient is \(\Kc=\alpha_1\alpha_2(B_2-B_1)/(\alpha_1B_2+\alpha_2B_1)\), with \(B_j\) the phase bulk modulus of Eq.~\eqref{eq:eos}.

For the \(x\)-directed Riemann problem the \(y\)-derivative terms are dropped and the system reduces to
\begin{equation}
\label{eq:quasi1d}
\partial_t U + \partial_x F(W) = N^u(W)\,\partial_x u + N^v(W)\,\partial_x v,
\end{equation}
with conservative variables
\begin{equation}
\label{eq:fullstate}
U=
\begin{bmatrix}
\alpha_1 & m_1 & m_2 & \rho u & \rho v & \rho E & \rho\tau_{xx} & \rho\tau_{yy} & \rho\tau_{xy}
\end{bmatrix}^T,
\end{equation}
conservative flux
\begin{equation}
\label{eq:flux}
F(W)=
\begin{bmatrix}
\alpha_1 u \\
m_1 u \\
m_2 u \\
\rho u^2 + p - \tau_{xx} \\
\rho uv - \tau_{xy} \\
(\rho E+p)u - u\tau_{xx} - v\tau_{xy} \\
\rho u\tau_{xx} \\
\rho u\tau_{yy} \\
\rho u\tau_{xy}
\end{bmatrix},
\end{equation}
and nonconservative-coefficient column vectors
\begin{equation}
\label{eq:nc}
N^u(W) =
\begin{bmatrix}
\alpha_1+\Kc \\
0 \\
0 \\
0 \\
0 \\
0 \\
\rho\bigl(\tfrac43 G+\tau_{xx}\bigr) \\
-\rho\bigl(\tfrac23 G+\tau_{yy}\bigr) \\
0
\end{bmatrix},
\qquad
N^v(W) =
\begin{bmatrix}
0 \\
0 \\
0 \\
0 \\
0 \\
0 \\
0 \\
2\rho\tau_{xy} \\
\rho\bigl(G+\tau_{xx}\bigr)
\end{bmatrix}.
\end{equation}
We call the vector \(U\) of Eq.~\eqref{eq:fullstate} the conservative variables, as opposed to the primitive variables \(W\) of Eq.~\eqref{eq:Wdef}, though the volume fraction and the stress components carry the nonconservative sources \(N^u, N^v\) of Eq.~\eqref{eq:nc} and are not conserved at the governing-equation level, a distinction made precise in \ref{sec:analysis}.

The \(y\)-directed subsystem follows by the permutation
\begin{equation}
\label{eq:yperm}
(u,\tau_{xx}) \leftrightarrow (v,\tau_{yy}),
\qquad
\tau_{xy}\ \text{invariant}.
\end{equation}

The three-dimensional system of~\ref{app:3d} adds the velocity component \(w\) and the stress components \(\tau_{xz}\), \(\tau_{yz}\), and \(\tau_{zz}\).
The two-dimensional system~\eqref{eq:full2d} is the in-plane form of the model, following Rodriguez and Johnsen~\cite{RodriguezJohnsen2019}: the strain rate~\eqref{eq:dij} keeps the three-dimensional deviatoric factor \(\tfrac13\), the out-of-plane stress \(\tau_{zz}\) is not tracked, and the elastic energy~\eqref{eq:caloric} is correspondingly in-plane.
In the linear regime the in-plane dynamics coincide with the plane-strain response, which serves as the semi-analytic reference in Section~\ref{subsec:s61-acoustic-scattering}.
Three dimensions add no new wave family.
For the \(x\)-directed quasi-one-dimensional subsystem, the tangential pairs \((v,\tau_{xy})\) and \((w,\tau_{xz})\) share the same shear-wave speed \(c_s=\sqrt{(G+\tau_{xx})/\rho}\) and are resolved independently but simultaneously by the inner HLLD waves.
The cross-stress \(\tau_{yz}\) is transported at the contact speed \(u\), with nonconservative corrections coupling it to both tangential pairs through the coefficients \(\tau_{xz}\) and \(\tau_{xy}\).
The \(y\)- and \(z\)-directed subsystems follow by index permutation, and the full three-dimensional system and directional mappings are given in~\ref{app:3d}.

\subsection{Characteristic structure}\label{subsec:char}

Introducing the primitive-like vector
\begin{equation}
\label{eq:Wdef}
W=
\begin{bmatrix}
\alpha_1 & m_1 & m_2 & u & v & p & \tau_{xx} & \tau_{yy} & \tau_{xy}
\end{bmatrix}^T,
\end{equation}
Eq.~\eqref{eq:quasi1d} may be recast in quasi-linear form,
\begin{equation}
\label{eq:quasilinear}
\partial_t W + \widetilde A(W)\,\partial_x W = 0,
\end{equation}
with the characteristic matrix
\begin{equation}
\label{eq:jacobian}
\widetilde A(W)=
\begin{bmatrix}
u & 0 & 0 & -\Kc & 0 & 0 & 0 & 0 & 0 \\
0 & u & 0 & m_1 & 0 & 0 & 0 & 0 & 0 \\
0 & 0 & u & m_2 & 0 & 0 & 0 & 0 & 0 \\
0 & 0 & 0 & u & 0 & 1/\rho & -1/\rho & 0 & 0 \\
0 & 0 & 0 & 0 & u & 0 & 0 & 0 & -1/\rho \\
0 & 0 & 0 & \rho a^2 & 0 & u & 0 & 0 & 0 \\
0 & 0 & 0 & -\left(\tfrac{4G}{3}+\tau_{xx}\right) & 0 & 0 & u & 0 & 0 \\
0 & 0 & 0 & \tfrac{2G}{3}+\tau_{yy} & -2\tau_{xy} & 0 & 0 & u & 0 \\
0 & 0 & 0 & 0 & -\left(G+\tau_{xx}\right) & 0 & 0 & 0 & u
\end{bmatrix}.
\end{equation}
This is the characteristic matrix used by the solver, following Rodriguez and Johnsen~\cite{RodriguezJohnsen2019}.
The tilde marks an approximation in the pressure row, quantified below.
In this quasi-linear form the nonconservative source of Eq.~\eqref{eq:quasi1d} is absorbed into the characteristic matrix.
The conservative variables \(U\) of~\eqref{eq:fullstate} and the primitive variables \(W\) of~\eqref{eq:Wdef} are two representations of the same state, related by the closure of Eqs.~\eqref{eq:Edef}--\eqref{eq:mixrules}.
Each expression uses whichever representation it requires.
The eigenvalues of \(\widetilde A\) are
\begin{equation}
\label{eq:eigs}
\lambda_1=\cdots=\lambda_5=u,
\qquad
\lambda_{s,\pm}=u\pm c_s,
\qquad
\lambda_{f,\pm}=u\pm c_f,
\end{equation}
with
\begin{equation}
\label{eq:wavespeeds}
c_s=\sqrt{\frac{G+\tau_{xx}}{\rho}},
\qquad
c_f=\sqrt{a^2+(4G/3+\tau_{xx})/\rho}.
\end{equation}
The pressure row of \(\widetilde A\) is an approximation.
With the elastic energy of Eq.~\eqref{eq:caloric} carried in the conserved energy, a direct chain-rule computation for a single material with constant \((\gamma,\pinf,G)\) gives the exact pressure row
\begin{equation}
\label{eq:exact-prow}
\partial_t p+u\,\partial_x p+\bigl[\gamma(p+\pinf)+b_u\bigr]\,\partial_x u+b_v\,\partial_x v=0,
\end{equation}
with
\begin{equation}
\label{eq:BuBv}
b_u=(\gamma-1)\left[-\frac{\tau_{xx}+\tau_{yy}}{3}+\frac{3\tau_{xx}^2-\tau_{yy}^2+2\tau_{xy}^2}{4G}\right],
\qquad
b_v=(\gamma-1)\,\frac{\tau_{xy}\,(\tau_{xx}+\tau_{yy})}{G}.
\end{equation}
Following Rodriguez and Johnsen~\cite{RodriguezJohnsen2019}, \(\widetilde A\) drops \(b_u\) and \(b_v\), on the basis that they are small relative to \(\gamma(p+\pinf)\) in the small-deformation regime targeted here.
The approximation affects only the pressure row.
The shear speed \(c_s\) is therefore exact.
The fast speed \(c_f\) is an estimate, with the exact single-material fast speed satisfying \(c_{f,\mathrm{exact}}^2=c_f^2+b_u/\rho\), so \(c_f\) may slightly differ from the exact fast speed.
For admissible elastic states \(G+\tau_{xx}>0\) and \(c_s\) is real.
A concise set of right eigenvectors, normalized to expose the variable jumps, is given in Table~\ref{tab:eigenvectors}.
The rows of the table correspond to the primitive components of \(W\) in the ordering of~\eqref{eq:Wdef} and each column is one eigenvector, so the nonzero entries of a column indicate directly which variables jump on the associated wave.
The shear and fast pairs are tabulated with \(\pm\) (resp.\ \(\mp\)) signs that correspond to the upper and lower eigenvalues \(\lambda_{s,\pm}=u\pm c_s\) and \(\lambda_{f,\pm}=u\pm c_f\).

\begin{table}[htbp]
\centering
\caption{Right eigenvectors of \(\widetilde A(W)\) in the primitive ordering~\eqref{eq:Wdef}. Each column is one eigenvector, denoted \(r_\bullet\); each row is one primitive variable. Signs \(\pm\) and \(\mp\) correspond to the upper and lower eigenvalues of the shear pair \(\lambda_{s,\pm}\) and the fast pair \(\lambda_{f,\pm}\).}
\label{tab:eigenvectors}
\renewcommand{\arraystretch}{1.3}
\setlength{\tabcolsep}{6pt}
\begin{tabular}{lccccccc}
\toprule
 & \multicolumn{5}{c}{Contact (\(\lambda=u\))} & Shear (\(\lambda_{s,\pm}\)) & Fast (\(\lambda_{f,\pm}\)) \\
\cmidrule(lr){2-6}\cmidrule(lr){7-7}\cmidrule(lr){8-8}
Variable & \(r_{\alpha}\) & \(r_{m_1}\) & \(r_{m_2}\) & \(r_{p_T}\) & \(r_{\tau_{yy}}\) & \(r_{s,\pm}\) & \(r_{f,\pm}\) \\
\midrule
\(\alpha_1\)  & 1 & 0 & 0 & 0 & 0 & 0 & \(\mp\,\Kc/c_f\) \\
\(m_1\)       & 0 & 1 & 0 & 0 & 0 & 0 & \(\pm\,m_1/c_f\) \\
\(m_2\)       & 0 & 0 & 1 & 0 & 0 & 0 & \(\pm\,m_2/c_f\) \\
\(u\)         & 0 & 0 & 0 & 0 & 0 & 0 & 1 \\
\(v\)         & 0 & 0 & 0 & 0 & 0 & 1 & 0 \\
\(p\)         & 0 & 0 & 0 & 1 & 0 & 0 & \(\pm\,\rho a^2/c_f\) \\
\(\tau_{xx}\) & 0 & 0 & 0 & 1 & 0 & 0 & \(\mp\,(\tfrac43 G+\tau_{xx})/c_f\) \\
\(\tau_{yy}\) & 0 & 0 & 0 & 0 & 1 & \(\mp\,2\tau_{xy}/c_s\) & \(\pm\,(\tfrac23 G+\tau_{yy})/c_f\) \\
\(\tau_{xy}\) & 0 & 0 & 0 & 0 & 0 & \(\mp\,\rho c_s\) & 0 \\
\bottomrule
\end{tabular}
\end{table}

Thus the five-fold contact family carries the material variables \((\alpha_1,m_1,m_2)\) and the transverse stress \(\tau_{yy}\).
Across it, \(p\) and \(\tau_{xx}\) jump while the total pressure \(p_T = p - \tau_{xx}\) stays invariant.
The shear pair at \(u\pm c_s\) carries the tangential velocity \(v\), the shear stress \(\tau_{xy}\), and, when \(\tau_{xy}\ne0\), a coupled jump in \(\tau_{yy}\).
The fast pair at \(u\pm c_f\) carries the longitudinal variables \((u,p,\tau_{xx})\), a coupled jump in \(\tau_{yy}\), and the material-compression response in \(\alpha_1,m_1,m_2\).
The shear eigenvectors depend explicitly on \(G+\tau_{xx}\) through the shear speed \(c_s\) of Eq.~\eqref{eq:wavespeeds}.
The wave speed and impedance are thus controlled by the same coefficient that appears in the nonconservative stress-production term of Eq.~\eqref{eq:nc}.
This is the algebraic reason that the HLLD shear subfan cannot be treated as a passive post-processing correction, since the nonconservative coefficient is part of the eigenstructure of a resolved wave family.

\subsection{Limiting structure}\label{subsec:limits}

Two limiting cases clarify the model's relationship to simpler systems and motivate the solver design.

In the fluid limit, with \(G=0\) and \(\bm{\tau}=\bm{0}\), the shear speed vanishes (\(c_s=0\)), the fast speed reduces to the acoustic speed (\(c_f=a\)), and the deviatoric stress equations are trivially satisfied.
The remaining system is the Kapila-type five-equation model \cite{Kapila2001} for compressible two-phase flow:
\begin{equation}
\label{eq:kapila}
\begin{gathered}
\partial_t\alpha_1+u\,\partial_x\alpha_1 = \Kc\,\partial_x u,
\qquad
\partial_t m_j + \partial_x(m_j u)=0,
\\[2pt]
\partial_t(\rho u)+\partial_x(\rho u^2+p)=0,
\qquad
\partial_t(\rho v)+\partial_x(\rho u v)=0,
\\[2pt]
\partial_t(\rho E)+\partial_x\bigl((\rho E+p)\,u\bigr)=0,
\end{gathered}
\end{equation}
with the mixture sound speed given by Wood's relation, Eq.~\eqref{eq:eos}, at \(G_j=0\).
The numerical method must reduce to a standard diffuse-interface update in this limit.

In the single-material elastic limit, with \(\alpha_1=1\) and \(\Kc=0\), the volume-fraction equation becomes trivial and the system reduces to a single-material hypoelastic solid.
The wave structure retains all five families, and the solver should behave as an HLLD solver for a single elastic material, with the anchoring retained because \(G+\tau_{xx}\) still varies from cell to cell.

The shear pair at \(u\pm c_s\) carries the tangential velocity \(v\), the shear stress \(\tau_{xy}\), and, when \(\tau_{xy}\ne0\), the coupled transverse stress \(\tau_{yy}\).
A three-wave HLLC solver does not resolve this pair.
It carries the tangential velocity as a single averaged state across the outer fast waves, much as HLL averages every field, and transports the shear stress with the contact.
The shear waves are then smeared over the fast-wave fan through the averaged velocity, rather than resolved at the slower shear speed, so HLLC remains more diffusive than HLLD on the shear waves.
The HLLD extension resolves the shear pair explicitly at the cost of a more complex star-state construction and the anchored dual-pass face evaluation discussed in Section~\ref{subsec:hlld}.

\FloatBarrier
\section{Numerical method}\label{sec:method}

The numerical method is organized around the three-tier treatment of the nonconservative products described in the introduction.
Section~\ref{subsec:ic-nc-update} states the interface-consistent nonconservative update, a conservative-form finite-volume update of the full state \(U\) that combines the standard HLL or HLLC numerical flux with an out-of-solver nonconservative source built from face-trace velocities supplied by the chosen wave fan.
The HLL baseline used for uniform comparisons follows in Section~\ref{subsec:hll}, and the HLLC single-pass member in Section~\ref{subsec:hllc}, whose out-of-solver source uses the HLLC face-trace velocities of Eq.~\eqref{eq:Psi-hllc}.
Section~\ref{subsec:hlld} develops the anchored HLLD production solver, in which the nonconservative coefficients enter the Riemann solve itself, and Section~\ref{subsec:reductions} records the reductions among these members.
Section~\ref{subsec:algorithms} summarizes the update algorithmically, and Section~\ref{subsec:cost} records the measured solver cost.
The construction is presented for the two-dimensional system of Section~\ref{sec:governing}.
The full three-dimensional system is given in~\ref{app:3d}.

\subsection{Interface-consistent nonconservative update}\label{subsec:ic-nc-update}

The HLL and HLLC schemes of Sections~\ref{subsec:hll}--\ref{subsec:hllc} update the conservative variables \(U\) of Eq.~\eqref{eq:fullstate} through a finite-volume scheme combining the standard HLL or HLLC numerical flux for the conservative-flux part \(\partial_x F(W)\) of Eq.~\eqref{eq:quasi1d} with an out-of-solver source carrying the nonconservative contribution from Eq.~\eqref{eq:nc}.
Cell \(i\)'s update reads
\begin{equation}
\label{eq:oos-update}
\begin{aligned}
 U_i^{n+1} = U_i^n
 &- \frac{\Delta t}{\Delta x}\bigl(\mathcal{F}^c_{i+1/2}-\mathcal{F}^c_{i-1/2}\bigr) \\
 &+ \frac{\Delta t}{\Delta x}\Bigl[N^u_i\bigl(\Psi_{u,i+1/2}-\Psi_{u,i-1/2}\bigr) + N^v_i\bigl(\Psi_{v,i+1/2}-\Psi_{v,i-1/2}\bigr)\Bigr],
\end{aligned}
\end{equation}
where \(\mathcal{F}^c\) is the HLL or HLLC numerical conservative flux (defined in Sections~\ref{subsec:hll}--\ref{subsec:hllc} below), \(N^u_i := N^u(W_i),\, N^v_i := N^v(W_i)\) are the conservative-form nonconservative-coefficient column vectors of Eq.~\eqref{eq:nc} evaluated at the cell-\(i\) primitive variables, and the face-trace velocities \(\Psi_u, \Psi_v\) are supplied by the wave fan of the chosen scheme: HLL traces in Section~\ref{subsec:hll}, Eq.~\eqref{eq:Psi-hll}; HLLC traces in Section~\ref{subsec:hllc}, Eq.~\eqref{eq:Psi-hllc}.

The face-trace velocities \(\Psi_u, \Psi_v\) are not independent reconstructed quantities but the velocity values the fan presents at the face, constructed for each scheme in Sections~\ref{subsec:hll} and~\ref{subsec:hllc}.

The construction is structurally the same idea as the familiar \(\Kc\,\partial_x u\) and \(\alpha\,\partial_x u\) treatments in Kapila-type five-equation schemes: a cell-centered coefficient multiplies a velocity gradient built from an interface-consistent velocity trace supplied by the Riemann fan.
The hypoelastic terms require the same structure for the coefficients \(G+\tau_{xx},\, \tfrac43 G+\tau_{xx},\, \tau_{xy}\) that appear as entries of \(N^u, N^v\) in Eq.~\eqref{eq:nc}.
The analytical statements of \ref{sec:analysis} concern Eq.~\eqref{eq:oos-update} with Eqs.~\eqref{eq:Psi-hll} and~\eqref{eq:Psi-hllc}, and the anchored face operator introduced in Section~\ref{subsec:hlld}.

The interface-consistent nonconservative update of Eq.~\eqref{eq:oos-update} expresses each face contribution to cell~\(i\) as an anchor-independent conservative flux \(\mathcal F^c\) plus an out-of-solver source built from face-trace velocities \((\Psi_u, \Psi_v)\).
In terms of an anchored face operator \(\mathcal R(W_L, W_R; \widehat W)\), this update corresponds to a face flux of the form
\begin{equation}
\label{eq:ic-representation}
\mathcal R(W_L, W_R; \widehat W)
= \mathcal F^c(W_L, W_R)
- N^u(\widehat W)\,\Psi_u(W_L, W_R)
- N^v(\widehat W)\,\Psi_v(W_L, W_R),
\end{equation}
with \(\mathcal F^c, \Psi_u, \Psi_v\) independent of the anchor \(\widehat W\).
We call the anchored face flux \emph{representable in interface-consistent form} on a branch when such a triple exists for every admissible anchor sampling that branch.
On branches that admit this form, the anchored dual-pass cell update collapses to the single-pass update of Eq.~\eqref{eq:oos-update} (Lemma~\ref{lem:dual-pass-to-bb}).
On branches that do not, no single-pass face representation of this form exists.
\ref{sec:analysis} gives a sufficient condition for a branch to admit no such representation (Proposition~\ref{prop:anchoring-criterion}) and shows that the active shear branches of anchored HLLD satisfy it.

\subsection{HLL baseline with interface-consistent traces}\label{subsec:hll}

The HLL member uses the two-wave fan \(S_L<S_R\) with the fast-wave estimates used throughout this paper,
\begin{equation}
\label{eq:SLSR}
S_L=\min\bigl(u_L-c_{f,L},\,u_R-c_{f,R}\bigr),
\qquad
S_R=\max\bigl(u_L+c_{f,L},\,u_R+c_{f,R}\bigr).
\end{equation}
As in Section~\ref{subsec:char}, these are signal-speed estimates of the reduced characteristic model.
For a scalar conservation law with conserved variable \(U\) and flux \(F\), \(\partial_t U+\partial_x F=0\), the HLL flux is
\begin{equation}
\label{eq:hll-generic}
\mathcal{F}_{\mathrm{HLL}}(U,F)=
\begin{cases}
F_L, & 0\le S_L,\\[2mm]
\dfrac{S_R F_L-S_L F_R+S_LS_R(U_R-U_L)}{S_R-S_L}, & S_L\le0\le S_R,\\[4mm]
F_R, & S_R\le0.
\end{cases}
\end{equation}
We use the standard interface-consistent velocity traces~\cite{JohnsenColonius2006}, obtained by setting the conserved variable to unity and the flux to the velocity in Eq.~\eqref{eq:hll-generic}, \(\mathcal{F}_{\mathrm{HLL}}(1,u)\) and \(\mathcal{F}_{\mathrm{HLL}}(1,v)\).
These traces give the velocity gradients in the nonconservative update the same HLL wave-speed weighting as the conservative flux.
The HLL traces are
\begin{equation}
\label{eq:Psi-hll}
\begin{aligned}
\Psi_u^{\mathrm{HLL}}=\mathcal{F}_{\mathrm{HLL}}(1,u)&=
\begin{cases}
u_L, & 0\le S_L,\\[1mm]
\dfrac{S_Ru_L-S_Lu_R}{S_R-S_L}, & S_L\le0\le S_R,\\[3mm]
u_R, & S_R\le0,
\end{cases}
\\[2mm]
\Psi_v^{\mathrm{HLL}}=\mathcal{F}_{\mathrm{HLL}}(1,v)&=
\begin{cases}
v_L, & 0\le S_L,\\[1mm]
\dfrac{S_Rv_L-S_Lv_R}{S_R-S_L}, & S_L\le0\le S_R,\\[3mm]
v_R, & S_R\le0.
\end{cases}
\end{aligned}
\end{equation}
The same traces serve the volume fraction and the stresses alike, supplying the velocity gradients in the columns \(N^u\) and \(N^v\) of Eq.~\eqref{eq:nc}.
The HLL conservative-core flux is obtained by applying Eq.~\eqref{eq:hll-generic} componentwise to the conservative variables and flux, and the conservative-form update is completed through Eq.~\eqref{eq:oos-update} with the velocity traces~\eqref{eq:Psi-hll}.\footnote{A finite-difference nonconservative update, in which \(\partial_x u\) is replaced by a centered cell-difference of \(u\), is equally usable for HLL in practice and is the form used by the original Rodriguez and Johnsen implementation in MFC: because HLL smears the material interface immediately through its two-wave dissipation, the sharp-contact pathology that forces interface-consistent transport for HLLC does not arise in the same way.
The present paper uses the interface-consistent form everywhere so that HLL, HLLC, and HLLD are compared within a common split-update framework.}

\subsection{Three-wave HLLC}\label{subsec:hllc}

The HLLC construction uses the three-wave fan
\begin{equation}
\label{eq:hllc-fan}
S_L < S_M < S_R,
\end{equation}
with outer bounds given by Eq.~\eqref{eq:SLSR}.
Introduce the total pressure (also called the normal traction)
\begin{equation}
\label{eq:pTdef}
p_T = p - \tau_{xx}
\end{equation}
and the left/right impedances
\begin{equation}
\label{eq:Adef}
A_L=\rho_L(S_L-u_L),
\qquad
A_R=\rho_R(S_R-u_R).
\end{equation}
Continuity of \(u\) and \(p_T\) across the contact gives the contact speed and common total pressure
\begin{equation}
\label{eq:SM-hllc}
S_M=
\frac{p_{T,R}-p_{T,L}+A_L u_L-A_R u_R}{A_L-A_R},
\end{equation}
\begin{equation}
\label{eq:Tstar-hllc}
p_T^* = p_{T,L} + A_L(S_M-u_L)
    = p_{T,R} + A_R(S_M-u_R).
\end{equation}
The compression factors are
\begin{equation}
\label{eq:xi}
\xi_K = \frac{S_K-u_K}{S_K-S_M}
= \frac{\rho_K^*}{\rho_K},
\qquad K\in\{L,R\}.
\end{equation}
The material-mass and normal-momentum star states are
\begin{equation}
\label{eq:star-masses}
m_{j,K}^*=m_{j,K}\,\xi_K,
\qquad
\rho_K^*=\rho_K\,\xi_K,
\qquad j=1,2,
\end{equation}
\begin{equation}
\label{eq:star-mom}
(\rho u)_K^* = \rho_K^* S_M.
\end{equation}
The shear waves of Table~\ref{tab:eigenvectors} carry the jumps in \(v\) and \(\tau_{xy}\).
HLLC omits this pair, so it must place these jumps on its three waves.
We let \(v\) jump across the outer waves \(S_L\) and \(S_R\) and remain continuous across \(S_M\).
Toro's HLLC integral consistency~\cite{ToroSpruceSpeares1994} then determines the common \(v^*\) by matching the fan-averaged tangential momentum \(\rho v\), with the conservative flux \(\rho uv-\tau_{xy}\) of Eq.~\eqref{eq:flux}, to its HLL average,
\begin{equation}
\label{eq:vstar-hllc}
v^*=\frac{A_R v_R-A_L v_L+\tau_{xy,R}-\tau_{xy,L}}{A_R-A_L}.
\end{equation}
The tangential momentum is \((\rho v)_K^* = \rho_K^*\,v^*\).
The stress production \(\rho(G+\tau_{xx})\partial_x v\) of Eq.~\eqref{eq:nc} is handled by the interface-consistent nonconservative source of Section~\ref{subsec:ic-nc-update} and does not enter the HLLC core.
The conservative shear-stress row, with flux \(\rho u\tau_{xy}\), is then a passive advection equation, and the outer-wave Rankine--Hugoniot conditions give
\begin{equation}
\label{eq:tauxystar-hllc}
\tau_{xy,K}^* = \tau_{xy,K}, \qquad K=L,R,
\end{equation}
so the transported shear stress jumps only at \(S_M\).
The energy star state, matching the energy flux \((\rho E + p)u - u\tau_{xx} - v\tau_{xy}\) of Eq.~\eqref{eq:flux}, is
\begin{equation}
\label{eq:star-energy}
(\rho E)_K^* =
\frac{(S_K-u_K)(\rho E)_K - p_{T,K} u_K + p_T^* S_M - v^*\bigl[\tau_{xy,K}-A_K(v^*-v_K)\bigr] + v_K\tau_{xy,K}}{S_K-S_M}.
\end{equation}
The trailing terms are the tangential shear-work contribution from the energy Rankine--Hugoniot balance.
The bracket is the shear traction that the tangential-momentum Rankine--Hugoniot condition gives across the same outer wave.
It is common to both sides of \(S_M\) by Eq.~\eqref{eq:vstar-hllc}, while the transported stress of Eq.~\eqref{eq:tauxystar-hllc} keeps its own value on each side.

Let
\begin{equation}
\label{eq:Ucstar}
U_K^{c,*} =
\begin{bmatrix}
\alpha_{1,K}^* & m_{1,K}^* & m_{2,K}^* & (\rho u)_K^* & (\rho v)_K^* & (\rho E)_K^* & (\rho\tau_{xx})_K^* & (\rho\tau_{yy})_K^* & (\rho\tau_{xy})_K^*
\end{bmatrix}^T.
\end{equation}
The HLLC numerical conservative-core flux is
\begin{equation}
\label{eq:hllc-coreflux}
\mathcal{F}^c_{\mathrm{HLLC}}=
\begin{cases}
F(W_L), & 0\le S_L,\\[2mm]
F(W_L)+S_L\bigl(U_L^{c,*}-U_L^c\bigr), & S_L\le 0\le S_M,\\[2mm]
F(W_R)+S_R\bigl(U_R^{c,*}-U_R^c\bigr), & S_M\le 0\le S_R,\\[2mm]
F(W_R), & S_R\le 0,
\end{cases}
\end{equation}
where \(F(W)\) is the conservative flux of Eq.~\eqref{eq:flux}, \(U_K^{c}\) collects the conservative variables of the reconstructed state \(W_K\), and \(U_K^{c,*}\) is the HLLC star state of Eq.~\eqref{eq:Ucstar}.

The volume-fraction and stress rows of the conservative-core flux are advected passively.
The volume fraction and all stresses reduce to the upwind conserved value scaled by the compression \(\xi_K\):
\begin{equation}
\label{eq:rhoz-star}
\begin{aligned}
\alpha_{1,K}^* &= \alpha_{1,K}\,\xi_K, &\qquad (\rho\tau_{xx})_K^* &= \rho_K^*\,\tau_{xx,K},\\
(\rho\tau_{yy})_K^* &= \rho_K^*\,\tau_{yy,K}, &\qquad (\rho\tau_{xy})_K^* &= \rho_K^*\,\tau_{xy,K}.
\end{aligned}
\end{equation}

The HLLC conservative flux of Eq.~\eqref{eq:hllc-coreflux} supplies the flux-divergence term of the update~\eqref{eq:oos-update}.
The out-of-solver source in~\eqref{eq:oos-update} needs, in addition, the face-trace velocities \(\Psi_u, \Psi_v\) of the HLLC fan.
As for HLL, \(\Psi_u\) is the HLLC flux of a unit passively advected scalar.
Its star value is \(\xi_K\) by Eq.~\eqref{eq:xi}, so its star flux is \(S_M\xi_K\) rather than the bare contact speed \(S_M\)~\cite{CoralicColonius2014}.
This is the velocity that carries the volume fraction in Eq.~\eqref{eq:hllc-coreflux}, \(\mathcal F^c_{\mathrm{HLLC}}[\alpha_1]=\alpha_{1,K}\Psi_u\) on each branch, so the conservative flux and the source of Eq.~\eqref{eq:oos-update} act on \(\alpha_1\) with one and the same face velocity.
The stress flux uses the same face velocity, \(\mathcal F^c_{\mathrm{HLLC}}[\rho\tau_{ij}]=\rho_K\tau_{ij,K}\Psi_u\).
With \(v^{*}\) from Eq.~\eqref{eq:vstar-hllc}, the traces are
\begin{equation}
\label{eq:Psi-hllc}
\begin{aligned}
\Psi_u^{\mathrm{HLLC}}&=
\begin{cases}
u_L, & 0\le S_L,\\[2mm]
S_M\,\xi_L, & S_L\le 0\le S_M,\\[2mm]
S_M\,\xi_R, & S_M\le 0\le S_R,\\[2mm]
u_R, & S_R\le 0,
\end{cases}
\\[2mm]
\Psi_v^{\mathrm{HLLC}}&=
\begin{cases}
v_L, & 0\le S_L,\\[2mm]
v^*, & S_L< 0< S_R,\\[2mm]
v_R, & S_R\le 0.
\end{cases}
\end{aligned}
\end{equation}

The HLLC scheme is then Eq.~\eqref{eq:oos-update} with the HLLC conservative flux of Eq.~\eqref{eq:hllc-coreflux} and the face-trace velocities of Eq.~\eqref{eq:Psi-hllc}.
This is a single-pass Riemann solve together with the out-of-solver nonconservative source, and no anchoring is required at this wave level.
Both schemes apply a single conservative flux $\mathcal{F}^c$ at each face, so they conserve $\{m_1, m_2, \rho u, \rho v, \rho E\}$ exactly.
The volume fraction and stresses are nonconservative because their governing equations are not in conservation-law form.

\subsection{Anchored HLLD solver}\label{subsec:hlld}

The production solver, the central construction of this paper, resolves the full five-wave structure and uses an anchored dual-pass update.\footnote{The term \emph{multi-pass} appears in some HLLD-type literature in an unrelated sense: L\'opez Ortega and coworkers~\cite{LopezOrtegaLombardiniPullinMeiron2014} and Wang, Adami, and Adams~\cite{WangAdamiAdams2024} use it to describe iterative refinement of a single intermediate state until the contact-wave speeds converge.
That construction yields one shared interface state, and hence one flux, per face.
The anchored dual-pass construction below is structurally different: it produces two distinct fluxes per face, anchored to the cell-centered states of the two adjacent cells, and each cell uses the flux anchored to its own state.}
At interface \(x_{i+1/2}\), reconstructed states \(W_{i+1/2}^L\) and \(W_{i+1/2}^R\) are fed to the approximate Riemann solver twice,
\begin{equation}
\label{eq:duplicate-solve}
\mathcal{F}_{i+1/2}^{\widehat L} = \mathcal{R}_{\mathrm{HLLD}}\bigl(W_{i+1/2}^L,W_{i+1/2}^R;\widehat W_i\bigr),
\qquad
\mathcal{F}_{i+1/2}^{\widehat R} = \mathcal{R}_{\mathrm{HLLD}}\bigl(W_{i+1/2}^L,W_{i+1/2}^R;\widehat W_{i+1}\bigr),
\end{equation}
where \(\widehat W_i\) and \(\widehat W_{i+1}\) are the cell-centered primitive variables on the left and right of the interface.

A hat marks a cell-centered value, taken from a cell rather than from a reconstructed face state, for example \(\widehat W_i\) or \(\widehat G^{\widehat D}\) as opposed to \(W_L\) or \(G_L\).
The anchor is tracked by the superscript \(\widehat D\in\{\widehat L,\widehat R\}\), which marks a quantity whose value depends on the anchor side, whether a cell-centered value such as \(\widehat G^{\widehat D}\) or an intermediate state such as the inner speed \(S_{L*}^{\widehat D}\) and the star volume fraction \(\alpha_{1,L}^{*,\widehat D}\).
A computed quantity without \(\widehat D\), for example \(S_M\), \(p_T^{*}\), or \(\rho_L^{*}\), is the same for both anchors.
\ref{sec:analysis} relies on this anchor-independence.

The first-order finite-volume update is
\begin{equation}
\label{eq:anchored-update}
U_i^{n+1}=U_i^n-\frac{\Delta t}{\Delta x}\left(\mathcal{F}_{i+1/2}^{\widehat L}-\mathcal{F}_{i-1/2}^{\widehat R}\right),
\end{equation}
that is, cell \(i\) uses the left-anchored flux on its right face and the right-anchored flux on its left face, so each face-flux contribution is evaluated with the cell's own material data.
Because \(\mathcal{F}_{i+1/2}^{\widehat L}\neq \mathcal{F}_{i+1/2}^{\widehat R}\) in general, the method is intentionally not pairwise conservative at an interface for some variables.
The subset of components for which pairwise conservation is retained is characterized in \ref{sec:analysis} (Proposition~\ref{prop:pairwise}).

\subsubsection{Anchored augmented flux}
The Rankine--Hugoniot balances that define the HLLD star states include an anchored contribution from the nonconservative terms of Eq.~\eqref{eq:quasi1d}.
With $N^u, N^v$ frozen at $\widehat W^{\widehat D}$, moving those terms to the LHS of~\eqref{eq:quasi1d}, with the expected sign flip, defines the \emph{anchored augmented flux}
\begin{equation}
\label{eq:anchored-flux}
F^{\widehat D}(W) := F(W) - N^u(\widehat W^{\widehat D})\,u(W) - N^v(\widehat W^{\widehat D})\,v(W),
\end{equation}
with $N^u, N^v$ the column vectors of Eq.~\eqref{eq:nc} and anchored coefficient values $\widehat\alpha_1^{\widehat D}, \widehat{\Kc}^{\widehat D}, \widehat\rho^{\widehat D}, \widehat G^{\widehat D}, \widehat\tau_{xx}^{\widehat D}, \widehat\tau_{yy}^{\widehat D}, \widehat\tau_{xy}^{\widehat D}$ in their evaluation.
Writing $[k]$ for the $k$th row of a state or flux vector, and since $N^u, N^v$ are nonzero only on the volume-fraction and stress rows, $F^{\widehat D}[k] = F[k]$ on the mass, momentum, and energy rows.
The $y$-sweep follows by~\eqref{eq:yperm}.
$F^{\widehat D}$ is the flux sampled by the HLLD star-state construction below.
The nonconservative terms thus enter through the flux alone, and the update~\eqref{eq:anchored-update} carries no out-of-solver source.

\subsubsection{Five-wave fan}
For fixed anchor \(\widehat D\), consider the five-wave fan
\begin{equation}
\label{eq:hlld-fan}
S_L < S_{L*}^{\widehat D} < S_M < S_{R*}^{\widehat D} < S_R,
\end{equation}
with star and double-star states
\begin{equation}
\label{eq:hlld-states}
U_L,
\quad
U_L^{*,\widehat D},
\quad
U_L^{**,\widehat D},
\quad
U_R^{**,\widehat D},
\quad
U_R^{*,\widehat D},
\quad
U_R.
\end{equation}
The outer bounds are those of Eq.~\eqref{eq:SLSR}.
Following the HLLD construction for ideal magnetohydrodynamics~\cite{MiyoshiKusano2005}, we take the wave-speed estimates to satisfy the ordering~\eqref{eq:hlld-fan}.
One degeneracy, a vanishing anchored impedance \(\widehat C^{\widehat D}\to 0\), arises in the fluid limit and is handled by the fallback of Section~\ref{subsec:reductions}.
The contact conditions
\begin{equation}
\label{eq:hlld-contact}
u_L^{*,\widehat D}=u_L^{**,\widehat D}=u_R^{**,\widehat D}=u_R^{*,\widehat D}=S_M,
\qquad
p_{T,L}^{*,\widehat D}=p_{T,L}^{**,\widehat D}=p_{T,R}^{**,\widehat D}=p_{T,R}^{*,\widehat D}=p_T^*,
\end{equation}
are inherited from the outer HLLC construction with \(S_M\) and \(p_T^*\) given by Eqs.~\eqref{eq:SM-hllc}--\eqref{eq:Tstar-hllc}.

\subsubsection{Outer star states}
For \(K\in\{L,R\}\), the anchored outer star states are
\begin{equation}
\label{eq:hlld-partialmasses}
m_{j,K}^{*}=m_{j,K}\,\xi_K,
\qquad
\rho_{K}^{*}=\rho_K\,\xi_K,
\qquad j=1,2,
\end{equation}
\begin{equation}
\label{eq:hlld-alpha}
\alpha_{1,K}^{*,\widehat D}=
\frac{\alpha_{1,K}(S_K-u_K)-\bigl(\widehat\alpha_1^{\widehat D}+\widehat\Kc^{\widehat D}\bigr)(S_M-u_K)}{S_K-S_M},
\end{equation}
\begin{equation}
\label{eq:hlld-tauxx}
\tau_{xx,K}^{*,\widehat D}=\tau_{xx,K}
- \frac{\widehat\rho^{\widehat D}\left(\tfrac43\widehat G^{\widehat D}+\widehat\tau_{xx}^{\widehat D}\right)(S_M-u_K)}{\rho_K(S_K-u_K)},
\end{equation}
\begin{equation}
\label{eq:hlld-tauyy}
\tau_{yy,K}^{*,\widehat D}=\tau_{yy,K}
+ \frac{\widehat\rho^{\widehat D}\left(\tfrac23\widehat G^{\widehat D}+\widehat\tau_{yy}^{\widehat D}\right)(S_M-u_K)}{\rho_K(S_K-u_K)},
\end{equation}
\begin{equation}
\label{eq:hlld-outer-tangential}
v_L^{*}=v_L,
\qquad
v_R^{*}=v_R,
\qquad
\tau_{xy,L}^{*}=\tau_{xy,L},
\qquad
\tau_{xy,R}^{*}=\tau_{xy,R},
\end{equation}
\begin{equation}
\label{eq:hlld-pstar}
p_K^{*,\widehat D}=p_T^*+\tau_{xx,K}^{*,\widehat D},
\end{equation}
\begin{equation}
\label{eq:hlld-energy-star}
(\rho E)_K^{*}=
\frac{(S_K-u_K)(\rho E)_K - p_{T,K} u_K + p_T^* S_M}{S_K-S_M}.
\end{equation}
Unlike the HLLC tangential velocity of Eq.~\eqref{eq:vstar-hllc}, which jumps across the outer waves, the HLLD outer star leaves the tangential fields unchanged across the outer wave \(S_K\).
They do not jump there, as their entries in the fast eigenvector vanish, as shown by Table~\ref{tab:eigenvectors}.
The tangential jump is resolved instead by the inner shear subfan.
The tangential work contribution of the HLLC energy star state cancels across the outer HLLD waves, so it does not enter Eq.~\eqref{eq:hlld-energy-star} and is instead accounted for in the double-star energy below.

\subsubsection{Inner shear subfan}
The shear subfan is what HLLD adds to HLLC, derived as the exact two-wave Riemann solution of the frozen tangential subsystem in the contact region.
Define the anchored shear impedance
\begin{equation}
\label{eq:Chat}
\widehat C^{\widehat D}=\sqrt{\widehat\rho^{\widehat D}\bigl(\widehat G^{\widehat D}+\widehat\tau_{xx}^{\widehat D}\bigr)}.
\end{equation}
Across the inner waves the material variables \(\alpha_1, m_j\), the density, the normal velocity \(u\), the pressure \(p\), and the normal stress \(\tau_{xx}\) do not jump and remain at their outer-star values.
This is the defining choice of the frozen tangential subsystem.
It matches the shear eigenvectors of the reduced matrix \(\widetilde A\), which carry jumps only in \(v\), \(\tau_{xy}\), and the coupled \(\tau_{yy}\) (Table~\ref{tab:eigenvectors}).
Density invariance gives
\begin{equation}
\label{eq:hlld-rho-dbl}
\rho_L^{**}=\rho_L^{*},
\qquad
\rho_R^{**}=\rho_R^{*},
\end{equation}
and the inner wave speeds are
\begin{equation}
\label{eq:inner-speeds}
S_{L*}^{\widehat D}=S_M-\frac{\widehat C^{\widehat D}}{\rho_L^{*}},
\qquad
S_{R*}^{\widehat D}=S_M+\frac{\widehat C^{\widehat D}}{\rho_R^{*}}.
\end{equation}
The inner speeds follow from the jump relations below and match the form of the ideal-MHD HLLD speeds~\cite{MiyoshiKusano2005}, with the anchored impedance in place of the Alfv\'en impedance.
The active variables across the inner waves are therefore \((v,\tau_{xy})\), together with the dependent shear-work increments to \(\tau_{yy}\) and the energy that follow from them.
The contact eigenvectors of Table~\ref{tab:eigenvectors} carry no jump in \(v\) or \(\tau_{xy}\), so the tangential state is continuous across the contact, \(v_L^{**}=v_R^{**}=:v^{**,\widehat D}\) and \(\tau_{xy,L}^{**}=\tau_{xy,R}^{**}=:\tau_{xy}^{**,\widehat D}\).

The tangential-momentum and shear-stress Rankine--Hugoniot relations across the left inner wave reduce, with the non-jumping variables at their outer-star values, to
\begin{equation}
\label{eq:inner-rh-L}
\begin{aligned}
\bigl(S_{L*}^{\widehat D}-S_M\bigr)\rho_L^{*}
\bigl(v^{**,\widehat D}-v_L\bigr)
+
\bigl(\tau_{xy}^{**,\widehat D}-\tau_{xy,L}\bigr)
&=0, \\
\bigl(S_{L*}^{\widehat D}-S_M\bigr)\rho_L^{*}
\bigl(\tau_{xy}^{**,\widehat D}-\tau_{xy,L}\bigr)
+
\bigl(\widehat C^{\widehat D}\bigr)^2
\bigl(v^{**,\widehat D}-v_L\bigr)
&=0,
\end{aligned}
\end{equation}
with the analogous pair across the right inner wave.
For a nontrivial jump, each pair in Eq.~\eqref{eq:inner-rh-L} forces \(\rho_K^{*}\bigl(S_{K*}^{\widehat D}-S_M\bigr)=\mp\widehat C^{\widehat D}\), recovering Eq.~\eqref{eq:inner-speeds}.
The relations are then equivalent to the matching of the two shear invariants
\begin{equation}
\label{eq:shear-invariants}
v_L - \frac{\tau_{xy,L}}{\widehat C^{\widehat D}}
=
v^{**,\widehat D} - \frac{\tau_{xy}^{**,\widehat D}}{\widehat C^{\widehat D}},
\qquad
v_R + \frac{\tau_{xy,R}}{\widehat C^{\widehat D}}
=
v^{**,\widehat D} + \frac{\tau_{xy}^{**,\widehat D}}{\widehat C^{\widehat D}},
\end{equation}
across the left and right inner waves respectively.
Solving the resulting \(2\times 2\) linear system yields the common inner state
\begin{equation}
\label{eq:vstarstar}
v^{**,\widehat D}=
\frac12\left(v_L+v_R+\frac{\tau_{xy,R}-\tau_{xy,L}}{\widehat C^{\widehat D}}\right),
\end{equation}
\begin{equation}
\label{eq:tauxystarstar}
\tau_{xy}^{**,\widehat D}=
\frac12\left(\tau_{xy,L}+\tau_{xy,R}+\widehat C^{\widehat D}(v_R-v_L)\right).
\end{equation}
The inner pair is therefore exact for the frozen tangential subsystem rather than an empirical correction to HLLC, carrying the anchored impedance \(\widehat C^{\widehat D}\).
Once \((v^{**,\widehat D},\tau_{xy}^{**,\widehat D})\) are fixed, the remaining inner Rankine--Hugoniot relations determine the transverse-stress and energy double-star increments uniquely.

Across the shear subfan, the material variables, \(u\), \(p\), and \(\tau_{xx}\) remain unchanged from the outer-star values:
\begin{equation}
\label{eq:const-shear}
\alpha_{1,K}^{**,\widehat D}=\alpha_{1,K}^{*,\widehat D},
\quad
m_{j,K}^{**}=m_{j,K}^{*},
\quad
p_K^{**,\widehat D}=p_K^{*,\widehat D},
\quad
\tau_{xx,K}^{**,\widehat D}=\tau_{xx,K}^{*,\widehat D},
\end{equation}
while the transverse stress and the energy receive the double-star increments
\begin{equation}
\label{eq:tauyy-Lstarstar}
\tau_{yy,L}^{**,\widehat D}=
\tau_{yy,L}^{*,\widehat D}
+\frac{2\widehat\rho^{\widehat D}\widehat\tau_{xy}^{\widehat D}}{\widehat C^{\widehat D}}
\left(v^{**,\widehat D}-v_L\right),
\end{equation}
\begin{equation}
\label{eq:tauyy-Rstarstar}
\tau_{yy,R}^{**,\widehat D}=
\tau_{yy,R}^{*,\widehat D}
-\frac{2\widehat\rho^{\widehat D}\widehat\tau_{xy}^{\widehat D}}{\widehat C^{\widehat D}}
\left(v^{**,\widehat D}-v_R\right),
\end{equation}
\begin{equation}
\label{eq:energy-Lstarstar}
(\rho E)_L^{**,\widehat D}=
(\rho E)_L^{*}
+\frac{\rho_L^{*}}{\widehat C^{\widehat D}}
\left(v^{**,\widehat D}\,\tau_{xy}^{**,\widehat D}-v_L\,\tau_{xy,L}\right),
\end{equation}
\begin{equation}
\label{eq:energy-Rstarstar}
(\rho E)_R^{**,\widehat D}=
(\rho E)_R^{*}
-\frac{\rho_R^{*}}{\widehat C^{\widehat D}}
\left(v^{**,\widehat D}\,\tau_{xy}^{**,\widehat D}-v_R\,\tau_{xy,R}\right).
\end{equation}

\subsubsection{Anchored sampled flux}\label{subsubsec:sampled-flux}
The conservative star vectors \(U_K^{*,\widehat D}\) and \(U_K^{**,\widehat D}\) are assembled from the scalar relations above, with each stress row assembled as \((\rho\tau_{\bullet})^{*}=\rho^{*}\,\tau_{\bullet}^{*}\), as in Eq.~\eqref{eq:rhoz-star}.
Their fluxes are defined recursively from the anchored augmented flux~\eqref{eq:anchored-flux} evaluated at the outer reconstructed states,
\begin{equation}
\label{eq:flux-outer}
\mathcal{F}_L^{\widehat D}:=F^{\widehat D}(W_L),
\qquad
\mathcal{F}_R^{\widehat D}:=F^{\widehat D}(W_R),
\end{equation}
through
\begin{equation}
\label{eq:flux-left-star}
\mathcal{F}_L^{*,\widehat D}=\mathcal{F}_L^{\widehat D}+S_L\bigl(U_L^{*,\widehat D}-U_L\bigr),
\qquad
\mathcal{F}_R^{*,\widehat D}=\mathcal{F}_R^{\widehat D}+S_R\bigl(U_R^{*,\widehat D}-U_R\bigr),
\end{equation}
\begin{equation}
\label{eq:flux-double-star}
\mathcal{F}_L^{**,\widehat D}=\mathcal{F}_L^{*,\widehat D}+S_{L*}^{\widehat D}\bigl(U_L^{**,\widehat D}-U_L^{*,\widehat D}\bigr),
\qquad
\mathcal{F}_R^{**,\widehat D}=\mathcal{F}_R^{*,\widehat D}+S_{R*}^{\widehat D}\bigl(U_R^{**,\widehat D}-U_R^{*,\widehat D}\bigr).
\end{equation}
Here \(U_L, U_R\) denote the conservative variables of the reconstructed outer states \(W_L, W_R\).
For a fixed anchor, the HLLD interface contribution is then
\begin{equation}
\label{eq:hlld-flux}
\mathcal{F}_{\mathrm{HLLD}}^{\widehat D}(W_L,W_R)=
\begin{cases}
\mathcal{F}_{L}^{\widehat D}, & 0\le S_L,\\[2mm]
\mathcal{F}_{L}^{*,\widehat D}, & S_L\le 0\le S_{L*}^{\widehat D},\\[2mm]
\mathcal{F}_{L}^{**,\widehat D}, & S_{L*}^{\widehat D}\le 0\le S_M,\\[2mm]
\mathcal{F}_{R}^{**,\widehat D}, & S_M\le 0\le S_{R*}^{\widehat D},\\[2mm]
\mathcal{F}_{R}^{*,\widehat D}, & S_{R*}^{\widehat D}\le 0\le S_R,\\[2mm]
\mathcal{F}_{R}^{\widehat D}, & S_R\le 0.
\end{cases}
\end{equation}
Combined with the dual-pass update~\eqref{eq:anchored-update}, Eq.~\eqref{eq:hlld-flux} defines the first-order anchored HLLD solver.

When the anchored shear impedance vanishes, the inner-wave formulas are singular in \(\bigl(\widehat C^{\widehat D}\bigr)^{-1}\), and the solver falls back to the three-wave form of Section~\ref{subsec:reductions} by collapsing the inner pair onto the contact, \(S_{L*}^{\widehat D}=S_{R*}^{\widehat D}=S_M\).

\begin{lemma}[Cell-update reduction]
\label{lem:dual-pass-to-bb}
Suppose at each face \(x_{i\pm 1/2}\) adjacent to cell \(i\) the anchored face flux \(\mathcal R(W_L, W_R; \widehat W)\) admits the interface-consistent representation of Eq.~\eqref{eq:ic-representation} with anchor-independent \((\mathcal F^c_{i\pm 1/2}, \Psi_{u, i\pm 1/2}, \Psi_{v, i\pm 1/2})\).
Then the anchored dual-pass cell update of Eq.~\eqref{eq:anchored-update} with anchor \(\widehat W_i = W_i\) coincides with the interface-consistent update of Eq.~\eqref{eq:oos-update}, so that the dual pass is not required on such branches.
\end{lemma}

\begin{proof}
Substituting Eq.~\eqref{eq:ic-representation} at both faces into Eq.~\eqref{eq:anchored-update} with \(\widehat W_i = W_i\), the anchor-independent \(\mathcal F^c\) contributions form the conservative-flux divergence, and the \(-N^u(\widehat W_i)\Psi_{u, i\pm 1/2} - N^v(\widehat W_i)\Psi_{v, i\pm 1/2}\) contributions assemble into the out-of-solver source of Eq.~\eqref{eq:oos-update}.
\end{proof}

\subsection{Reductions}\label{subsec:reductions}

The anchored HLLD fan admits a three-wave reduction when the reconstructed tangential velocity and shear stress agree on both sides of the face.
With \(v_L=v_R\) and \(\tau_{xy,L}=\tau_{xy,R}\), Eqs.~\eqref{eq:vstarstar}--\eqref{eq:tauxystarstar} give \(v^{**,\widehat D}=v_L=v_R\) and \(\tau_{xy}^{**,\widehat D}=\tau_{xy,L}=\tau_{xy,R}\), and the double-star increments~\eqref{eq:tauyy-Lstarstar}--\eqref{eq:energy-Rstarstar} for \(\tau_{yy}\) and the energy vanish identically.
The inner shear pair is then inactive, and every sampled branch of Eq.~\eqref{eq:hlld-flux} takes the corresponding value of the reduced HLLD flux
\begin{equation}
\label{eq:degenerate-anchor-flux}
\mathcal{F}_{\mathrm{HLLD,reduction}}^{\widehat D}(W_L,W_R)=
\begin{cases}
\mathcal{F}_{L}^{\widehat D}, & 0\le S_L,\\[2mm]
\mathcal{F}_{L}^{*,\widehat D}, & S_L\le 0\le S_M,\\[2mm]
\mathcal{F}_{R}^{*,\widehat D}, & S_M\mathrel{<}0\le S_R,\\[2mm]
\mathcal{F}_{R}^{\widehat D}, & S_R\le 0.
\end{cases}
\end{equation}
The outer-star branches retain the HLLD tangential velocities \(v_K^{*}=v_K\) (Eq.~\eqref{eq:hlld-outer-tangential}), whereas HLLC uses the common velocity \(v^{*}\) of Eq.~\eqref{eq:vstar-hllc}.
Both constructions retain the transported stress \(\tau_{xy,K}^{*}=\tau_{xy,K}\).
For general tangential data, this three-wave prescription therefore differs from the hypoelastic HLLC construction.

For fluid states with \(\bm{\tau}=\bm{0}\), taking \(G\to0\) removes the shear contribution.
The resulting three-wave flux gives a conventional single-pass fluid HLLC update with the interface-consistent nonconservative source for the volume-fraction terms.
The inactive shear pair leaves the outer stars to transport the tangential velocities per side, \(v_K^*=v_K\).
In contrast, the hypoelastic HLLC member of Section~\ref{subsec:hllc} does not reduce to this fluid HLLC, as it keeps the HLL-like averaged tangential velocity across the fast waves.

For general face data, the solver also adopts Eq.~\eqref{eq:degenerate-anchor-flux} as a fallback whenever \(\widehat G+\widehat\tau_{xx}\) vanishes to floating-point tolerance in an anchored evaluation.
An anchored state with \(\widehat G+\widehat\tau_{xx}\) negative beyond this tolerance, where the shear speed of Eq.~\eqref{eq:wavespeeds} ceases to be real, takes the HLL flux instead.\footnote{The HLL flux is also taken if the assembled inner speeds fall outside the outer fan, \(S_{L*}^{\widehat D}<S_L\) or \(S_{R*}^{\widehat D}>S_R\).
This did not occur in any of the computations of Sections~\ref{sec:results-1d} and~\ref{sec:results-2d}.}
The local structural properties established in \ref{sec:analysis} concern the anchored HLLD construction with these HLL fallbacks inactive.

\subsection{Algorithmic summary}\label{subsec:algorithms}

The directional update of Sections~\ref{subsec:ic-nc-update} to~\ref{subsec:hlld} is summarized by Algorithms~\ref{alg:hll-hllc} and~\ref{alg:anchored-hlld}.
Both compute the semi-discrete \emph{residual} \(\mathcal L\), the stage right-hand side of \(\mathrm dU/\mathrm dt=\mathcal L\).
Starting from the stage state \(U^n\), each accumulates the \(x\)-direction contribution of one Runge--Kutta stage into \(\mathcal L\) (zeroed before the first pass).
Algorithms~\ref{alg:hll-hllc} and~\ref{alg:anchored-hlld} share Phase~1 and differ only in Phase~2 and in how Phase~3 uses its outputs.

\begin{algorithm}[htbp]
\caption{HLL/HLLC interface-consistent update: \(x\)-direction pass of one Runge--Kutta stage.}
\label{alg:hll-hllc}
\begin{algorithmic}[1]
\Phase{Phase 1 -- reconstruct (once per face):}{}
\ForAll{\(x\)-face \(i+\tfrac12\)}
  \State \(W^L_{i+1/2},\;W^R_{i+1/2}\gets\textsc{Reconstruct}(U^n)\)
\EndFor
\Statex
\Phase{Phase 2 -- one Riemann solve per face:}{face-trace velocities, Eqs.~\eqref{eq:Psi-hll},\,\eqref{eq:Psi-hllc}}
\ForAll{\(x\)-face \(i+\tfrac12\)}
  \State \(\mathcal F^c_{i+1/2},\,\Psi_{u,i+1/2},\,\Psi_{v,i+1/2}\gets\textsc{HLL}/\textsc{HLLC}\bigl(W^L_{i+1/2},\,W^R_{i+1/2}\bigr)\)
\EndFor
\Statex
\Phase{Phase 3 -- accumulate residual:}{conservative flux + out-of-solver nonconservative source, Eqs.~\eqref{eq:oos-update},\,\eqref{eq:nc}}
\ForAll{cell \(i\)}
  \State \(N^u_i,\,N^v_i\gets N^u(W_i),\,N^v(W_i)\)
  \State \(\mathcal L_i\gets\mathcal L_i-\dfrac{1}{\Delta x}\bigl(\mathcal F^c_{i+1/2}-\mathcal F^c_{i-1/2}\bigr)\)
  \State \(\phantom{\mathcal L_i\gets{}}+\dfrac{1}{\Delta x}\Bigl[N^u_i\bigl(\Psi_{u,i+1/2}-\Psi_{u,i-1/2}\bigr)+N^v_i\bigl(\Psi_{v,i+1/2}-\Psi_{v,i-1/2}\bigr)\Bigr]\)
\EndFor
\end{algorithmic}
\end{algorithm}

\begin{algorithm}[htbp]
\caption{Anchored HLLD update: \(x\)-direction pass of one Runge--Kutta stage.}
\label{alg:anchored-hlld}
\begin{algorithmic}[1]
\Phase{Phase 1 -- reconstruct (once per face):}{shared with Alg.~\ref{alg:hll-hllc}}
\ForAll{\(x\)-face \(i+\tfrac12\)}
  \State \(W^L_{i+1/2},\;W^R_{i+1/2}\gets\textsc{Reconstruct}(U^n)\)
\EndFor
\Statex
\PhaseTwoLine{Phase 2 -- dual Riemann solve per face:}{nonconservative terms built into the solve, two anchored fluxes per face, Eq.~\eqref{eq:hlld-flux}}
\ForAll{\(x\)-face \(i+\tfrac12\)}
  \State \(\mathcal F^{\widehat L}_{i+1/2}\gets\textsc{HLLD}\bigl(W^L_{i+1/2},\,W^R_{i+1/2};\,\widehat W_i\bigr)\) \Comment{left cell anchor}
  \State \(\mathcal F^{\widehat R}_{i+1/2}\gets\textsc{HLLD}\bigl(W^L_{i+1/2},\,W^R_{i+1/2};\,\widehat W_{i+1}\bigr)\) \Comment{right cell anchor}
\EndFor
\Statex
\Phase{Phase 3 -- accumulate residual:}{each cell uses the flux anchored to its own state, Eq.~\eqref{eq:anchored-update}}
\ForAll{cell \(i\)}
  \State \(\mathcal L_i\gets\mathcal L_i-\dfrac{1}{\Delta x}\bigl(\mathcal F^{\widehat L}_{i+1/2}-\mathcal F^{\widehat R}_{i-1/2}\bigr)\)
\EndFor
\end{algorithmic}
\end{algorithm}

A full stage repeats the same three phases for the \(y\)-pass under the permutation~\eqref{eq:yperm}, accumulating into the \emph{same} \(\mathcal L\).
In three dimensions a \(z\)-pass follows, and for the HLL/HLLC update each pass additionally carries the \(w\)-velocity trace \(N^w\,\Psi_w\), with \(N^w\) given in~\ref{app:3d} and \(\Psi_w\) defined analogously to \(\Psi_v\).
The SSP-RK3 integrator then advances \(U\) from the accumulated residual.
For forward Euler, \(U_i^{n+1}=U_i^n+\Delta t\,\mathcal L_i\).

\subsection{Computational cost}\label{subsec:cost}

Relative to HLL, one might expect the anchored HLLD update to more than double the cost, since the five-wave fan involves more arithmetic per solve and the dual pass evaluates it twice at each face.
The measured increase stays below a factor of two on GPUs and far below it on CPU nodes.
The reconstruction, primitive conversions, and halo exchange are shared across the three solvers and performed once per stage, and the operations shared between the two anchored evaluations are likewise performed once, so only the anchor-dependent arithmetic is computed twice.
On application benchmarks in MFC's GPU- and CPU-parallel configurations~\cite{MFC2026,Radhakrishnan2024MFC}, measured for a three-dimensional water--epoxy problem on Frontier (AMD MI250X GPUs) and Bridges-2 (AMD EPYC CPU nodes), the per-stage HLLD/HLL wall-time ratio was $1.7$--$1.9$ on 1--128 GPUs and $1.1$--$1.2$ on 1--32 CPU nodes.
The cost increase is workload and machine dependent rather than a universal factor.

\subsection{AI-assisted computational workflow}\label{subsec:ai-workflow}

ChatGPT and Codex (OpenAI) and Claude Code (Anthropic) were used to support software development, code review, and post-processing.
These tools assisted with implementing and refining features within the open-source Multi-component Flow Code (MFC), identifying potential software issues, and developing scripts to process simulation outputs and generate scientific visualizations.
AI-assisted code modifications and analysis scripts were reviewed and verified by the authors prior to use.

\FloatBarrier
\section{Verification and solver comparison in one dimension}
\label{sec:results-1d}

These one-dimensional cases serve three roles: to \emph{verify} the schemes against exact and reference solutions, to \emph{compare} HLL, HLLC, and anchored HLLD on a common footing, and to \emph{demonstrate} the analytical propositions of \ref{sec:analysis}.
Exact solutions or closed-form predictions are available for many of these cases, so accuracy and convergence can be measured quantitatively.

HLL results throughout the results sections (this and following sections) follow the Rodriguez and Johnsen method~\cite{RodriguezJohnsen2019} implemented in the production MFC code~\cite{SprattThesis,MFC2026}, and serve as the reference baseline, together with analytic solutions where available.
HLLC and anchored HLLD are introduced in the present work.
The MFC solver~\cite{MFC2026} is used for all results in this paper.

\subsection{Common setup}
\label{subsec:results-1d-common}

All cases use the stiffened-gas equation of state extended with an elastic shear modulus $G$ (the hypoelastic constitutive model of Section~\ref{sec:governing}).
The materials used in the suite are listed in Table~\ref{tab:materials}.

\begin{table}[H]
\centering
\caption{Material catalog. \emph{Water} and \emph{epoxy} are the two
primary materials. The remaining two (\emph{epoxy-stiff},
\emph{water-stiff}) appear in specific cases.}
\label{tab:materials}
\begin{tabular}{l c c c c}
\toprule
Material      & $\gamma$ & $\pinf$ (Pa)       & $\rho$ (kg/m$^3$) & $G$ (Pa)            \\
\midrule
water         & $4.4$    & $6.0\times 10^8$   & $1000$            & $0$                 \\
epoxy         & $2.94$   & $2.53\times 10^9$  & $1185$            & $1.5\times 10^9$    \\
epoxy-stiff   & $2.94$   & $2.53\times 10^9$  & $1185$            & $6.0\times 10^9$    \\
water-stiff   & $4.4$    & $6.0\times 10^8$   & $1000$            & $1.0\times 10^9$    \\
\bottomrule
\end{tabular}
\end{table}

First-order upwind reconstruction uses forward Euler time stepping, second-order MUSCL uses the van Leer limiter with SSP-RK3, and WENO5 uses SSP-RK3.
Upwind is used primarily where a clear cross-solver comparison is the priority.
All cases are run with a fixed $\Delta t$ at a maximum CFL of 0.5, unless otherwise stated.
All quasi-one-dimensional cases (one-dimensional discretization with tangential and transverse variables $v$, $\tau_{xy}$, and $\tau_{yy}$) are implemented on a two-dimensional strip uniform in $y$.
The rotated-eigenmode case (Figure~\ref{fig:eigenmode-conv}(c)) uses a fully two-dimensional grid.
The eigenmode and drift cases are periodic in $x$, and the remaining cases use zero-gradient extrapolation at the $x$ boundaries.
The test configurations are listed in Table~\ref{tab:verification-suite}.

\begin{table}[H]
\centering
\caption{Test-suite configurations.
}
\label{tab:verification-suite}
\setlength{\tabcolsep}{3.5pt}
\begin{tabular}{l l l l}
\toprule
\S & Test & Materials & Reconstruction \\
\midrule
\ref{subsec:eigenmode-conv}     & Eigenmode convergence   & (a,b) epoxy            & MUSCL; WENO5 \\
                                &                         & (c) epoxy; (d) water/epoxy         & MUSCL             \\
\ref{subsubsec:acoustic-fresnel}& Acoustic Fresnel        & water/epoxy     & MUSCL \\
\ref{subsubsec:contact}         & Contact preservation    & epoxy/epoxy-stiff   & MUSCL \\
\ref{subsubsec:rodriguez}       & Riemann problem         & water-stiff         & Upwind; MUSCL \\
\ref{subsubsec:shear-release}   & Shear release           & water/epoxy     & Upwind; MUSCL \\
\ref{subsubsec:shear-fresnel}   & Shear Fresnel           & epoxy/epoxy-stiff   & MUSCL \\
\ref{subsubsec:drift-smooth}    & Drift convergence (smooth)   & water/epoxy        & Upwind; MUSCL \\
\ref{subsubsec:drift-discont}   & Drift convergence (discont.) & water/epoxy; epoxy & MUSCL \\
\bottomrule
\end{tabular}
\end{table}

\subsection{Smooth eigenmode convergence}
\label{subsec:eigenmode-conv}

Four smooth-IC eigenmode tests measure the spatial convergence of the discretized hypoelastic system by propagating an eigenwave across a periodic domain.
Each test initializes its eigenmode at a velocity amplitude of $10^{-3}$~m/s.
These tests are the numerical counterpart of Propositions~\ref{prop:smooth-hllc} (HLL/HLLC) and~\ref{prop:smooth-hlld} (HLLD), showing the solution error converging at the reconstruction order~$r$ on smooth solutions.
The $\theta = 30^\circ$ rotated two-dimensional shear case uses a rectangular domain, periodic in both directions, with $L_y = \sqrt{3}\,L_x$ so the rotated shear eigenmode is exactly periodic.
Cells are near-square with $N_y/N_x = 8/5$ and $\Delta y \approx 1.08\,\Delta x$.
The $45^\circ$ case, shown alongside in Figure~\ref{fig:eigenmode-conv}(c), uses a square grid with $L_y = L_x$.
The bulk-mixture shear case applies the eigenmode to a smoothly varying epoxy volume fraction $\alpha_1(x) = 0.5 + 0.4\sin(2\pi x)$ in water, a genuinely mixed-cell setting with no pure-material region.
The WENO5 simulations use CFL\,$= 0.05$ instead of $0.5$ so that the fifth-order WENO5 spatial error dominates the third-order SSP-RK3 temporal error at the resolutions shown.
The final times are fractions of an eigenmode period for the acoustic, shear, and bulk-mixture cases ($35.4$, $88.6$, and $177\,\mu$s) and one period for the rotated cases ($770$ and $628\,\mu$s at $\theta=30^\circ$ and $45^\circ$).

The implementation uses the finite-volume cell average itself as the cell-centered state $W_i$ and anchor $\widehat W_i$ of Propositions~\ref{prop:smooth-hllc} and~\ref{prop:smooth-hlld}, with no order-$r$ quadrature on the averages.
The cell average is second-order accurate as a point value, which meets the hypothesis on $W_i$ for first-order upwind and MUSCL but not for WENO5.
For WENO5 the resulting contribution to the residual is quadratic in the eigenmode amplitude, since the coefficient deviation and the gradients it multiplies are each proportional to it.
At the amplitudes used here this contribution lies below the measured error, and the fifth-order rates below are unaffected.

The reported error is the $L_1$ deviation of the active primitive: for Figure~\ref{fig:eigenmode-conv}(a)--(c) it is measured against the analytic eigenmode.
For (d), which has no closed-form solution on the smoothly varying mixture, it is against a high-resolution HLLD reference ($N = 1600$).
The curves change by at most $2\%$ when the high-resolution reference is computed with HLL or HLLC instead.

\begin{figure}[htbp]
\centering
\begin{subfigure}[t]{0.32\textwidth}
  \centering
  \includegraphics[width=\textwidth]{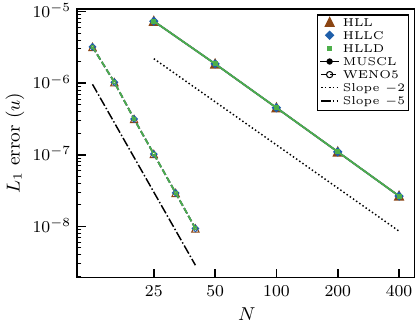}
  \caption{Acoustic (longitudinal)}
\end{subfigure}%
\hspace{0.01\textwidth}%
\begin{subfigure}[t]{0.32\textwidth}
  \centering
  \includegraphics[width=\textwidth]{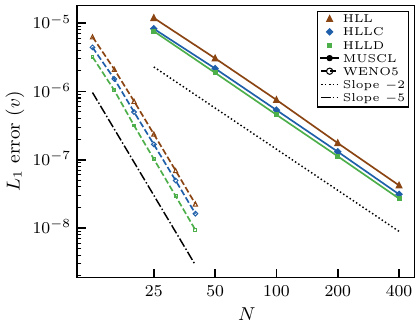}
  \caption{Shear}
\end{subfigure}%
\\[6pt]%
\begin{subfigure}[t]{0.32\textwidth}
  \centering
  \includegraphics[width=\textwidth]{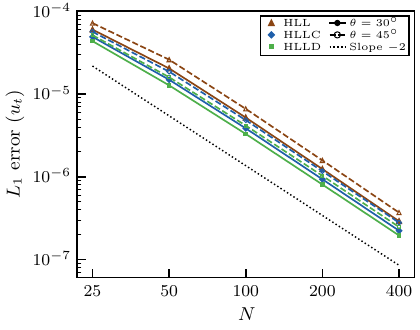}
  \caption{Rotated 2D shear}
\end{subfigure}%
\hspace{0.01\textwidth}%
\begin{subfigure}[t]{0.32\textwidth}
  \centering
  \includegraphics[width=\textwidth]{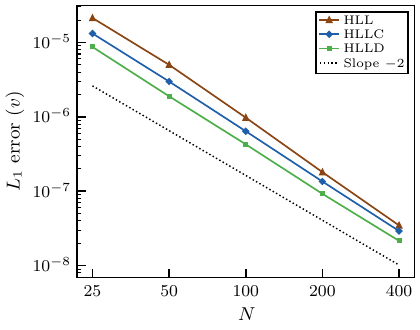}
  \caption{Bulk-mixture shear}
\end{subfigure}
\caption{Smooth eigenmode convergence: single-material epoxy (a) acoustic, (b) shear, (c) rotated 2D shear, 
and (d) bulk-mixture shear on smoothly varying $\alpha_1(x)$ between
epoxy and water. Panel (c) shows both $\theta=30^\circ$ and $\theta=45^\circ$. $N$ is points per wavelength.}
\label{fig:eigenmode-conv}
\end{figure}

HLL, HLLC, and anchored HLLD all recover the predicted reconstruction order~$r$ (Figure~\ref{fig:eigenmode-conv}): second order with MUSCL ($r=2$) and fifth order with WENO5 ($r=5$).
The rotated two-dimensional shear (c) and bulk-mixture (d) cases confirm that MUSCL retains second order under grid rotation and within a smoothly varying mixture, where $G$ (and hence the HLLD anchor) varies in space.
On the three shear-driven cases (b)--(d), HLLD reaches a lower error than HLL and HLLC at the same order, and on the longitudinal case (a) the three solvers coincide.

\subsection{Inactive-shear cases}
\label{subsec:degenerate}

In the next two cases, the elastic shear contribution that distinguishes HLLD from HLLC is inactive, and HLLD results match HLLC.

\subsubsection{Acoustic Fresnel transmission}
\label{subsubsec:acoustic-fresnel}

A Gaussian acoustic pressure pulse of peak amplitude $10^3$~Pa transmits across a sharp water--epoxy interface.
On the unit domain the interface sits at $x=0.5$, and the pulse $\delta p = 10^3\exp\bigl(-((x-0.2)/0.02)^2\bigr)$~Pa is initialized centered at $x_0=0.2$ with the matched right-moving velocity and density perturbations $\delta u=\delta p/(\rho c_f)$ and $\delta\rho=\delta p/c_f^2$.
The right-moving pulse is initiated in the water, which cannot support shear.
Figure~\ref{fig:acoustic-fresnel}(b) measures the $L_1$ error against the analytic solution using the pressure reflection and transmission coefficients $R = (Z_2 - Z_1)/(Z_2 + Z_1)$ and $T=2Z_2a_2^2/[(Z_1+Z_2)c_{f,2}^2]$, with $Z_j=\rho_jc_{f,j}$.
The final time is $250\,\mu$s.

\begin{figure}[htbp]
\centering
\begin{subfigure}[t]{0.32\textwidth}
  \centering
  \includegraphics[width=\textwidth]{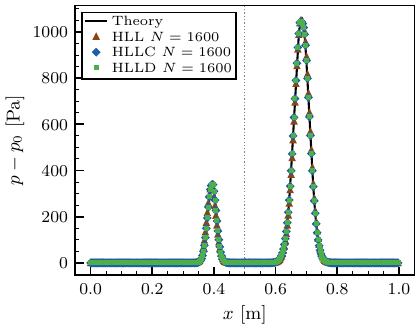}
  \caption{Waveform (every fifth point)}
\end{subfigure}%
\hspace{0.01\textwidth}%
\begin{subfigure}[t]{0.32\textwidth}
  \centering
  \includegraphics[width=\textwidth]{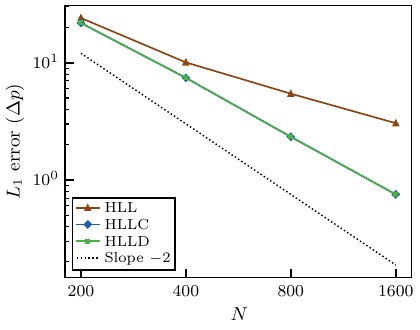}
  \caption{Convergence}
\end{subfigure}
\caption{Acoustic Fresnel transmission.}
\label{fig:acoustic-fresnel}
\end{figure}

Because the pulse is launched in the fluid, no shear is excited in either phase: the shear stress $\tau_{xy}$ and tangential velocity $v$ remain identically zero in both phases, and the HLLD and HLLC solutions coincide accordingly.
In Figure~\ref{fig:acoustic-fresnel}(b), HLLC and HLLD converge near second order, at a measured slope of about $1.6$ at the resolutions shown, whereas HLL converges sublinearly because it smears the material interface.

\subsubsection{Solid--solid material contact}
\label{subsubsec:contact}

This case demonstrates Proposition~\ref{prop:contact} on a contact discontinuity between two solids advecting at uniform velocity.
Following the proposition assumptions, four variables are uniform across the interface: $u = 100$~m/s, $v = 50$~m/s, $\tau_{xy} = 10^3$~Pa, and $p_T \equiv p - \tau_{xx} = 10^5$~Pa.
The material parameters and the remaining stress components are discontinuous across the interface: left side (epoxy) has $G = 1.5\times 10^9$~Pa, $p = 1.1\times 10^5$~Pa, $\tau_{xx} = 10^4$~Pa, $\tau_{yy} = -5\times 10^3$~Pa; right side (epoxy-stiff) has $G = 6.0\times 10^9$~Pa, $p = 10^5$~Pa, $\tau_{xx} = \tau_{yy} = 0$.
Both phases share the density $\rho = 1185$~kg/m$^3$.
The final time is $5.0$~ms.

\begin{figure}[htbp]
\centering
\begin{subfigure}[t]{0.32\textwidth}
  \centering
  \includegraphics[width=\textwidth]{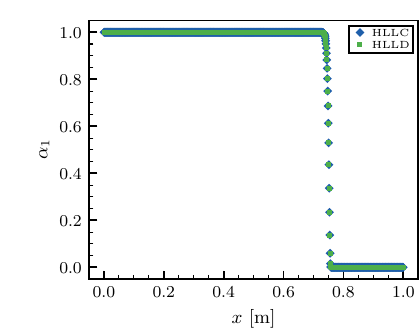}
  \caption{$\alpha_1$}
\end{subfigure}%
\hspace{0.01\textwidth}%
\begin{subfigure}[t]{0.32\textwidth}
  \centering
  \includegraphics[width=\textwidth]{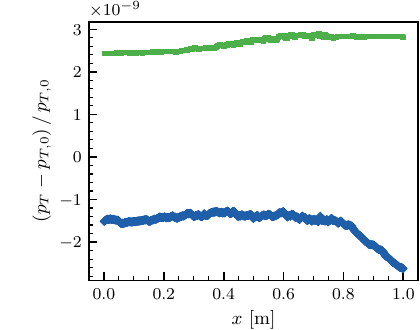}
  \caption{$(p_T - p_{T,0})/p_{T,0}$}
\end{subfigure}%
\\[6pt]%
\begin{subfigure}[t]{0.32\textwidth}
  \centering
  \includegraphics[width=\textwidth]{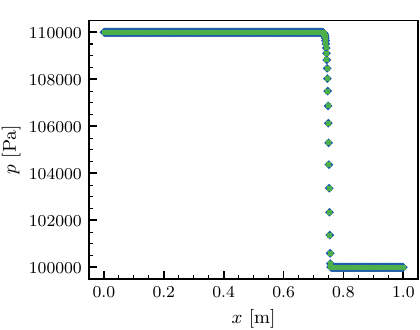}
  \caption{$p$}
\end{subfigure}%
\hspace{0.01\textwidth}%
\begin{subfigure}[t]{0.32\textwidth}
  \centering
  \includegraphics[width=\textwidth]{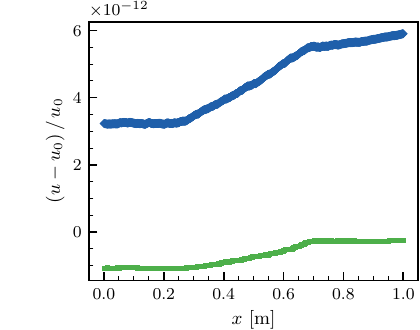}
  \caption{$(u - u_0)/u_0$}
\end{subfigure}
\caption{Solid--solid material contact: Proposition~\ref{prop:contact} demonstration. HLLC and anchored HLLD, the two members the proposition concerns, are shown.}
\label{fig:contact}
\end{figure}

Because the two phases have equal density, the density/normal-stress term of Proposition~\ref{prop:contact} vanishes, and the four values of the proposition are preserved across the advecting contact up to the elastic-energy blend residual of Eq.~\eqref{eq:pT-onestep} (Figure~\ref{fig:contact}): panels (b) and (d) show the total pressure $p_T$ and normal velocity $u$ remaining constant, and the unplotted tangential velocity $v$ and shear stress $\tau_{xy}$ behave the same way.
The velocities $u$ and $v$ are held to machine precision, with relative deviations below $10^{-11}$.
The stress-derived quantities $\tau_{xy}$ and $p_T$ are held near the recovery floor, of order $10^{-9}$ in relative terms for $p_T$, caused by both the conserved-to-primitive cancellation of the stiffened-gas closure and the elastic-energy blend residual quantified in Proposition~\ref{prop:contact}.
The pressure itself is not invariant, as demonstrated by panel (c), since the preserved quantity is the total pressure $p_T = p - \tau_{xx}$.
HLLC and HLLD coincide here, as the contact carries no $\tau_{xy}$ jump for the inner pair to act on.

\subsection{Shear-active comparison cases}
\label{subsec:shear-active}

In the next three cases, the shear waves are active, so we can compare how the three solvers resolve the shear, longitudinal, and contact waves.

\subsubsection{Rodriguez and Johnsen Riemann problem}
\label{subsubsec:rodriguez}

This case tests a single-material Riemann problem from Fig.~7 of Rodriguez and Johnsen (2019)~\cite{RodriguezJohnsen2019}, adapted from Gavrilyuk, Favrie, and Saurel (2008)~\cite{GavrilyukFavrieSaurel2008}.
The initial condition has tangential velocities of $\pm 100$~m/s, a pressure ratio of $10^5$ ($p_L = 10^{10}$~Pa, $p_R = 10^5$~Pa), and zero initial stress.
The case drives large stresses relative to the shear modulus, so it functions as a severe nonlinear stress test at the edge of the small-deformation regime, following its use in~\cite{RodriguezJohnsen2019}.
The final time is $64\,\mu$s.

\begin{figure}[htbp]
\centering
\begin{subfigure}[t]{0.32\textwidth}
  \centering
  \includegraphics[width=\textwidth]{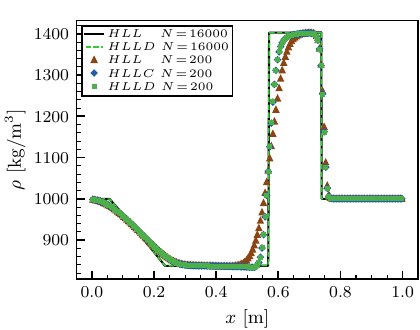}
  \caption{$\rho$}
\end{subfigure}%
\hspace{0.01\textwidth}%
\begin{subfigure}[t]{0.32\textwidth}
  \centering
  \includegraphics[width=\textwidth]{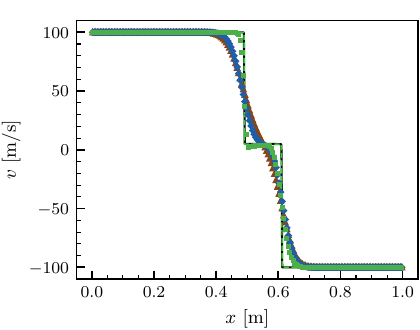}
  \caption{$v$}
\end{subfigure}%
\hspace{0.01\textwidth}%
\begin{subfigure}[t]{0.32\textwidth}
  \centering
  \includegraphics[width=\textwidth]{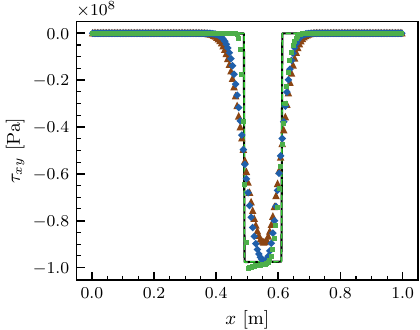}
  \caption{$\tau_{xy}$}
\end{subfigure}
\caption{Rodriguez and Johnsen (2019) Fig.~7 configuration. Markers show the $N = 200$ solutions and lines the $N = 16{,}000$ HLL and HLLD solutions.}
\label{fig:rodriguez}
\end{figure}

The low-resolution $N = 200$ results (Figure~\ref{fig:rodriguez}) show how the solvers resolve the waves differently.
In the density panel (a), HLLC and HLLD nearly coincide and capture the contact discontinuity more sharply than HLL.
In the tangential-velocity (b) and $\tau_{xy}$ (c) panels, which carry the shear waves, HLLD resolves the shear waves much more sharply than HLLC, which is sharper than HLL.
Thus HLLD resolves both the contact and the shear waves, HLLC the contact but not the shear, and HLL neither sharply.

In the density panel the $N = 16{,}000$ HLL and HLLD solutions overlay closely and both agree with~\cite{RodriguezJohnsen2019}, consistent with the schemes approaching a common limit for this problem despite their differing low-resolution profiles.

\subsubsection{Shear release wave}
\label{subsubsec:shear-release}

A sharp water--epoxy interface sits at $x=0.25$ on the unit domain, with an initial discontinuity in the shear stress ($\tau_{xy} = 0$ in water, $\tau_{xy,0} = 10^4$~Pa in epoxy) and uniform tangential $v_0 = 10$~m/s on both sides.
Theoretically, the release wave advects rightwards at the shear-wave speed $c_s = \sqrt{G/\rho}$ while remaining sharp, producing a velocity jump $\Delta v = \tau_{xy,0} / (\rho c_s)$.
The final time is $50\,\mu$s.

\begin{figure}[htbp]
\centering
\includegraphics[width=\textwidth]{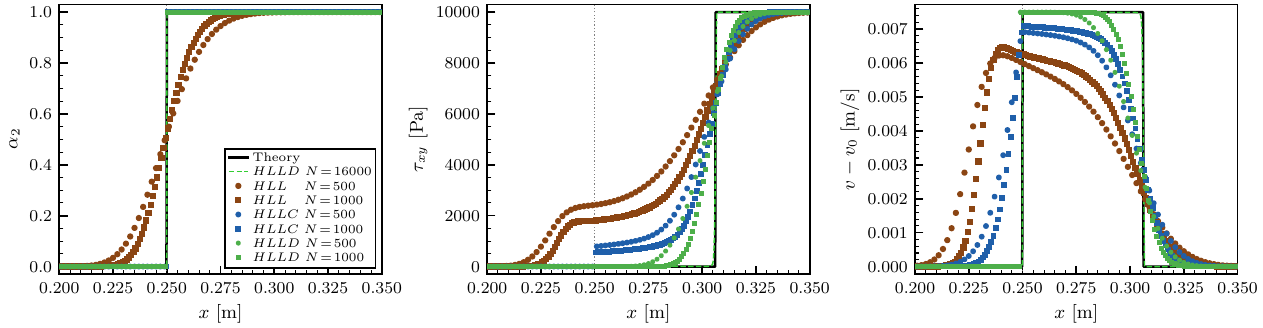}
\caption{Shear release at a water--epoxy interface.
Panels overlay the theoretical jump (\emph{Theory}), solutions at $N=500$ and $N=1000$, and an HLLD MUSCL solution at $N=16{,}000$.
The dotted line in the $\tau_{xy}$ panel marks the initial interface at $x=0.25$.}
\label{fig:shear-release}
\end{figure}

The three solvers differ strongly in this case (Figure~\ref{fig:shear-release}).
In the $\alpha_2$ panel, HLLC and anchored HLLD keep the material interface within one cell, while HLL smears it.
In the $\tau_{xy}$ panel, HLL smears the shear front and carries stress across the interface into the water.
HLLC also diffuses the shear front over many cells, but its stress drops to zero at the interface, across the contact that it resolves.
Anchored HLLD diffuses the shear much less, so the stress drops to zero across the release front, before it reaches the interface.
In the $v - v_0$ panel, HLL diffuses the velocity jump into the water, where it forms a smooth peak.
HLLC drops sharply once across the interface, with much less spreading into the water than HLL.
Anchored HLLD steps from $v_0$ to $v_0 + \Delta v$ across a single cell.
An HLLD solution computed with MUSCL at $N = 16{,}000$ demonstrates convergence to the theoretical jump $\Delta v = \tau_{xy,0}/(\rho c_s)$.

\subsubsection{Shear Fresnel transmission}
\label{subsubsec:shear-fresnel}

This case is the shear-wave analog of the acoustic Fresnel case: a Gaussian tangential-velocity pulse of peak amplitude $v = 1$~m/s, $v=\exp\bigl(-((x-0.2)/0.02)^2\bigr)$~m/s on the unit domain, with the shear stress matched for a right-moving wave, $\tau_{xy} = -\rho c_s v$ (peak $\approx 1.3\times 10^6$~Pa in the epoxy), transmits across an interface at $x=0.5$ between two solids of different $G$.
Convergence is measured against the analytic shear-Fresnel coefficients with $Z_s = \rho c_s$.
The final time is $450\,\mu$s.

\begin{figure}[htbp]
\centering
\begin{subfigure}[t]{0.32\textwidth}
  \centering
  \includegraphics[width=\textwidth]{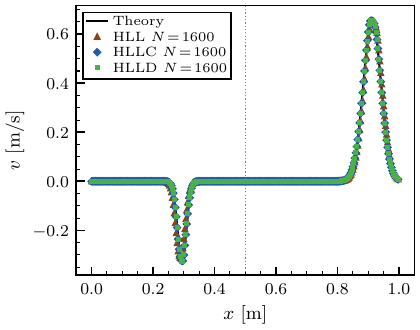}
  \caption{Waveform (every fifth point)}
\end{subfigure}%
\hspace{0.01\textwidth}%
\begin{subfigure}[t]{0.32\textwidth}
  \centering
  \includegraphics[width=\textwidth]{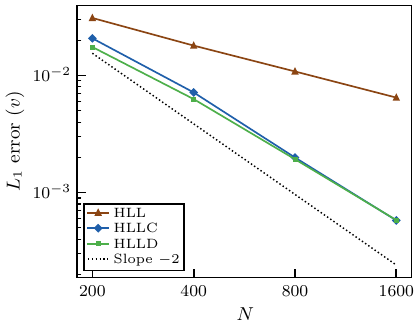}
  \caption{Convergence}
\end{subfigure}
\caption{Shear Fresnel transmission.}
\label{fig:shear-fresnel}
\end{figure}

As in the acoustic case, the transmitted and reflected shear amplitudes (Figure~\ref{fig:shear-fresnel}(a)) match the analytic Fresnel prediction, and panel (b) shows HLLC and HLLD are close to second order while HLL converges sublinearly due to interface smearing.

\subsection{Conservation drift convergence}
\label{subsec:drift-rate}

The cases below demonstrate that the drift converges at least at the reconstruction's face-jump order on smooth fields and at first order in the discontinuous cases.

\subsubsection{Smooth fields}
\label{subsubsec:drift-smooth}

This case demonstrates Proposition~\ref{prop:drift-rate} on smooth fields with face-jump orders $q = 1$ (upwind) and $q = 2$ (MUSCL).
The drift in tangential momentum and total energy is measured as the largest deviation of each spatial integral from its initial value over the run.

The smooth initial condition applies a tangential-velocity perturbation $v = 0.1\sin(2\pi x + 0.7)$ to the epoxy/water mixture $\alpha_1(x) = 0.5 + 0.4\sin(2\pi x)$ of Section~\ref{subsec:eigenmode-conv}, with zero initial stress and uniform pressure.
The final time is $4.5$~ms.

\begin{figure}[htbp]
\centering
\begin{subfigure}[t]{0.32\textwidth}
  \centering
  \includegraphics[width=\textwidth]{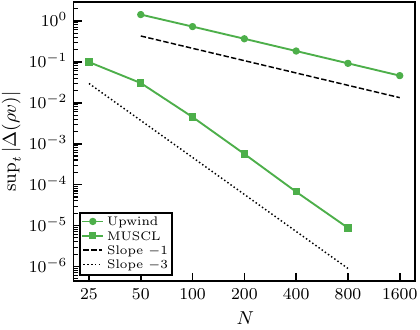}
  \caption{Tangential momentum}
\end{subfigure}%
\hspace{0.01\textwidth}%
\begin{subfigure}[t]{0.32\textwidth}
  \centering
  \includegraphics[width=\textwidth]{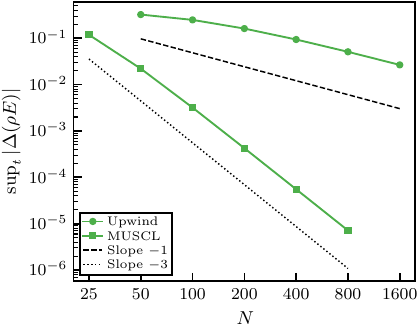}
  \caption{Total energy}
\end{subfigure}
\caption{HLLD conservation drift in tangential momentum and total energy on smooth fields.}
\label{fig:drift-rate}
\end{figure}

Both tangential momentum and total energy converge at least at the rates guaranteed by Proposition~\ref{prop:drift-rate}, with slope $-1$ for upwind and about $-3$ for MUSCL (Figure~\ref{fig:drift-rate}).

\subsubsection{Discontinuous fields}
\label{subsubsec:drift-discont}

The following two cases demonstrate first-order drift convergence in the presence of material interfaces and shocks.
Both use MUSCL on the unit periodic domain with zero initial stress and $p = 10^5$~Pa, and run to $t = 4.5$~ms.

The contact case has stationary material interfaces at $x = 0.25$ and $0.75$, with the epoxy volume fraction $\alpha_1 = 0.9$ for $0.25 \leq x \leq 0.75$ and $0.1$ otherwise.
The epoxy/water materials and tangential perturbation are unchanged from the smooth case of Section~\ref{subsubsec:drift-smooth}.

The shock case starts from uniform epoxy with $u = -0.2c_{f,0}\cos(2\pi x)$ and $v = 10\sin(2\pi x + 0.7)$~m/s, where $c_{f,0}$ is the initial longitudinal wave speed.
The normal velocity profile steepens into two shocks that propagate through the sheared field.

\begin{figure}[htbp]
\centering
\begin{subfigure}[t]{0.32\textwidth}
  \centering
  \includegraphics[width=\textwidth]{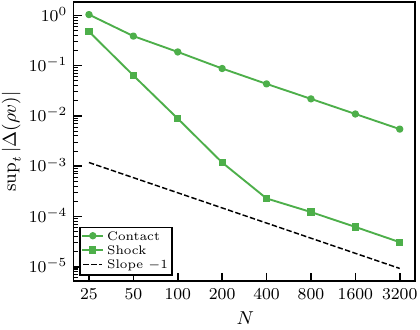}
  \caption{Tangential momentum}
\end{subfigure}%
\hspace{0.01\textwidth}%
\begin{subfigure}[t]{0.32\textwidth}
  \centering
  \includegraphics[width=\textwidth]{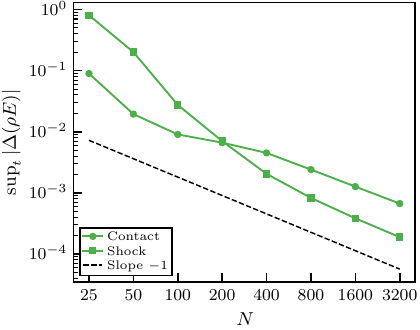}
  \caption{Total energy}
\end{subfigure}
\caption{HLLD conservation drift in tangential momentum and total energy for stationary material contacts (\emph{Contact}) and shocks propagating through a sheared field (\emph{Shock}).}
\label{fig:drift-discont}
\end{figure}

Both tangential momentum and total energy approach the slope $-1$ on the fine grids in Figure~\ref{fig:drift-discont}.
In the proof of Proposition~\ref{prop:drift-rate}, each face contribution is bounded by the magnitude of the anchor-impedance difference multiplied by a weighted sum of $|\Delta v|$ and $|\Delta\tau|$.
On smooth data, the tangential jumps are $\mathcal O(h^q)$ and the impedance difference is $\mathcal O(h)$.
Summing over $\mathcal O(h^{-1})$ faces gives an $\mathcal O(h^q)$ drift.
Material interfaces and shocks can produce $\mathcal O(1)$ impedance differences, while shear waves can carry $\mathcal O(1)$ tangential jumps.
For separated one-dimensional waves, an $\mathcal O(1)$ jump in one factor can be accompanied by an $\mathcal O(h)$ variation in the other.
For discontinuities spanning a bounded number of faces, the resulting drift is $\mathcal O(h)$.

\FloatBarrier
\section{Production-scale demonstration in two and three dimensions}
\label{sec:results-2d}

This section presents production-scale computations (high-resolution, fifth-order WENO) on three configurations: transient acoustic scattering by an elastic epoxy cylinder with a semi-analytic reference (Section~\ref{subsec:s61-acoustic-scattering}), a two-dimensional shock--cylinder interaction (Section~\ref{subsec:s62-shock-cylinder}), and a three-dimensional shock--ellipsoid interaction (Section~\ref{subsec:s63-3d-ellipsoid}) that engages a fully coupled three-dimensional stress field.
The transient-cylinder case is compared with a semi-analytic elastic-cylinder reference, while the shock--cylinder and shock--ellipsoid cases have no exact solution.
For the shock--cylinder case, the comparison rests on the stress response, its sensitivity to resolution, and cross-solver differences, quantified through the integrated diagnostics below.
HLL, following the Rodriguez and Johnsen method~\cite{RodriguezJohnsen2019}, is the reference baseline for HLLC and HLLD.
The shock--cylinder case also compares qualitatively with the analogous problem of Rodriguez and Johnsen.

\subsection{Common setup}
\label{subsec:results-2d-common}

All cases use the water and epoxy materials of Section~\ref{sec:results-1d} (Table~\ref{tab:materials}): water with $G = 0$ and epoxy with $G = 1.5\times 10^9$~Pa.
The spatial discretization is fifth-order WENO with mapped weights~\cite{HenrickAslamPowers2005} and the time integrator is SSP-RK3 throughout, at a CFL of approximately $0.5$.
Both two-dimensional cases compare HLL, HLLC, and anchored HLLD.
The transient-cylinder case uses characteristic non-reflecting boundaries~\cite{Thompson1987,Thompson1990} on all sides, and the shock--cylinder case uses the same treatment away from its symmetry plane on $y = 0$.
The three-dimensional case is a single anchored HLLD production run, included as a capability demonstration rather than a solver comparison.

\subsection{Transient acoustic scattering by an epoxy cylinder}
\label{subsec:s61-acoustic-scattering}

A planar Gaussian acoustic pulse propagates through water and scatters from a single epoxy cylinder.
The domain is $[0,1.5]\times[0,1.0]$~m, and the cylinder has radius $0.15$~m and center $(0.75,0.5)$.
The source is a planar right-going Gaussian pulse injected as a source term~\cite{MaedaColonius2017} near the left face at $x=0.05$, spanning the full height of the domain, with peak amplitude $1.0$~Pa, temporal width $\sigma_t=2\times 10^{-5}$~s, and delay $8\times10^{-5}$~s.
The case is run to $t=1000\,\mu$s, and results are shown on the $1200\times800$ and $2400\times1600$ grids.
HLL is initialized with a sharp interface to isolate solver-induced broadening, while HLLC and anchored HLLD use the one-cell regularized interface.
Even with this sharper initial data, HLL broadens the material transition into a wide mixed layer by the final time (Figure~\ref{fig:s61-alpha}).
Anchored HLLD preserves the one-cell initial interface, and the right panel of Figure~\ref{fig:s61-alpha} shows its change in $\alpha_2$ over the run to be of order $10^{-8}$.
HLLC (not shown) behaves the same, consistent with the role of the contact wave in preserving the water--epoxy interface.

This geometry has a semi-analytic plane-strain elastic-cylinder reference, obtained by writing the incident Gaussian pulse as a Fourier superposition of plane waves~\cite{Faran1951,Huang1970Transient,Geers1972Transient,VekslerKorsunski1980}.
For each wavenumber, the water pressure is expanded as the incident field plus outgoing Hankel modes, and the cylinder displacement is expanded with longitudinal and shear elastic potentials.
The modal coefficients are found by enforcing continuity of normal displacement, normal-stress balance, and zero tangential traction at the cylinder surface.
The summed modes give the thermodynamic pressure and deviatoric stresses used in the comparisons.

\begin{figure}[htbp]
\centering
\includegraphics[width=\textwidth]{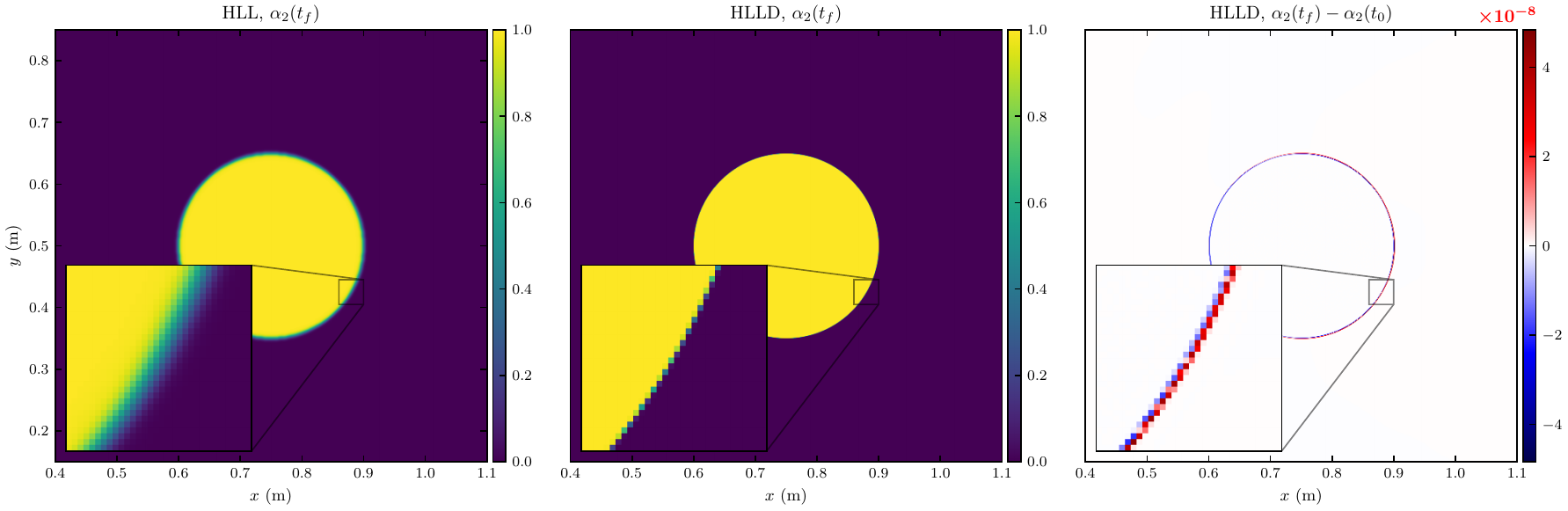}
\caption{Volume fraction $\alpha_2$ (epoxy) for the transient cylinder case on the $1200\times800$ grid.
The left and center panels show the final-time HLL and anchored HLLD fields.
The right panel shows the anchored HLLD change in $\alpha_2$ from the initial to the final time, $\alpha_2(t_f)-\alpha_2(t_0)$.
Insets zoom the lower-right cylinder interface.}
\label{fig:s61-alpha}
\end{figure}

\begin{figure}[htbp]
\centering
\includegraphics[width=\textwidth]{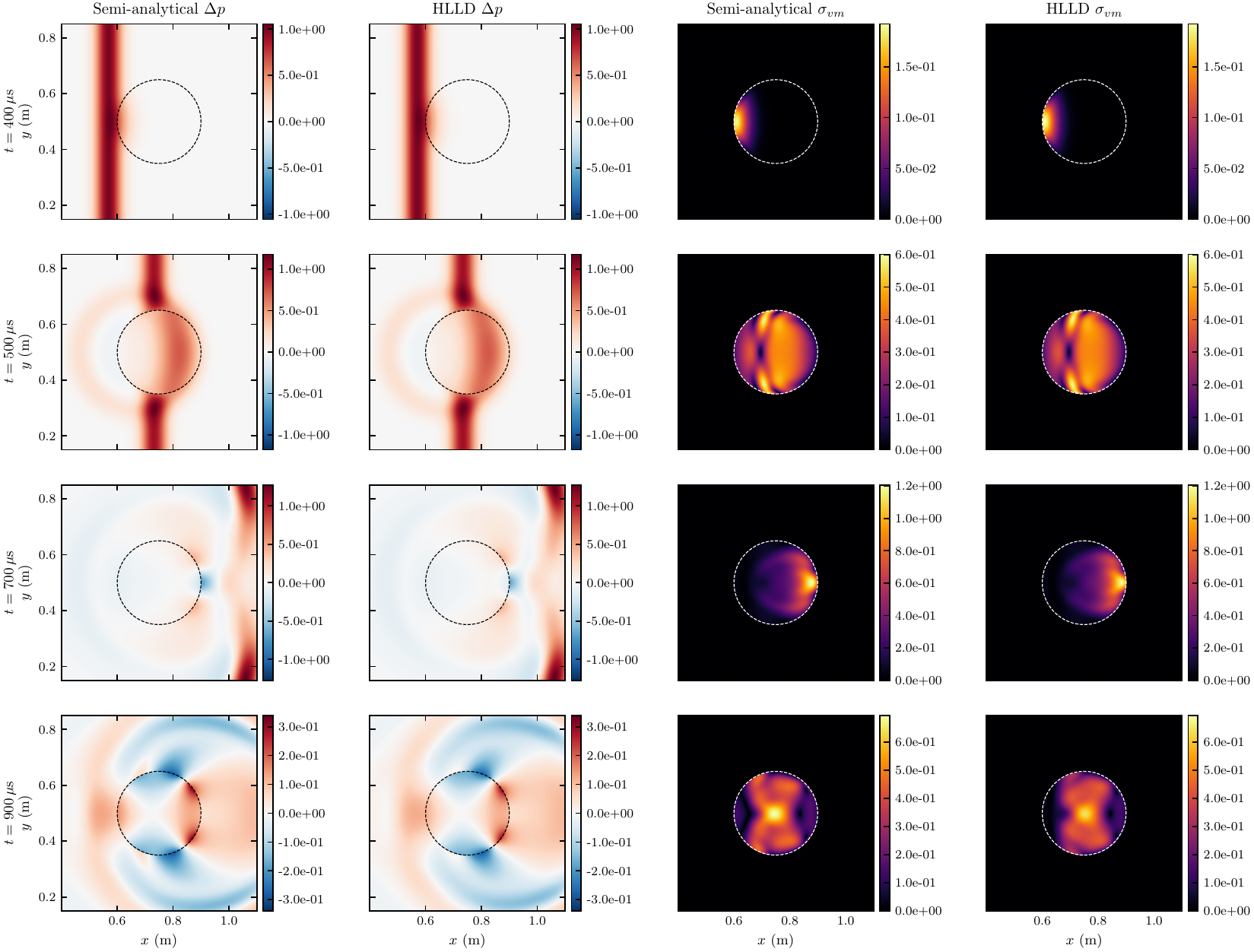}
\caption{Transient cylinder scattering on the finest, $2400\times1600$ grid.
Rows show $t = 400$, $500$, $700$, and $900\,\mu$s.
Panels show the window $0.40\le x\le 1.10$~m, $0.15\le y\le 0.85$~m of the full $[0,1.5]\times[0,1]$~m domain.
Columns compare the semi-analytic reference and anchored HLLD for pressure perturbation $\Delta p$ and von Mises stress $\sigma_{vm}$, the in-plane form $\sigma_{vm}=\sqrt{\tau_{xx}^2-\tau_{xx}\tau_{yy}+\tau_{yy}^2+3\tau_{xy}^2}$ used for all two-dimensional cases.
Dashed contours mark the cylinder interface.}
\label{fig:s61-contour}
\end{figure}

Figure~\ref{fig:s61-contour} compares the fields at representative times.
Anchored HLLD matches the reference pressure for the wave transmitted through the cylinder and for the scattered wave in the water.
The von Mises panels show the shear response confined to the solid, in qualitative agreement with the reference at each time.

\begin{figure}[htbp]
\centering
\includegraphics[width=\textwidth]{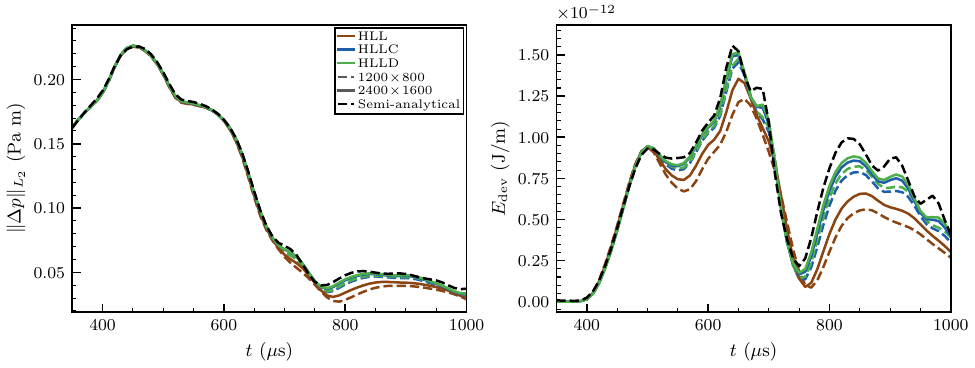}
\caption{Time evolution of $L_2|\Delta p|$ (left) and of the in-plane deviatoric strain energy $E_{\mathrm{dev}}(t) = \int_{\text{solid}} (\tau_{xx}^2 + \tau_{yy}^2 + 2\tau_{xy}^2)/(4G)\,\mathrm{d}A$ (right) for the transient cylinder case on the $1200\times800$ (dashed) and $2400\times1600$ (solid) grids.
The black dashed line is the semi-analytic elastic-cylinder reference.
$L_2|\Delta p|$ is evaluated over the cylinder-centered disc $r\le0.35$~m with an identical cell mask for simulation and reference, and the reference is defined on the $350$--$1000\,\mu$s comparison window.}
\label{fig:s61-histories}
\end{figure}

The integrated diagnostics compare the three solvers against the semi-analytic reference in time (Figure~\ref{fig:s61-histories}).
During pulse transit, $L_2|\Delta p|$ is dominated by the transmitted field inside the cylinder, and HLLC and anchored HLLD are already close to the reference over the plotted resolutions.
HLL has a much larger pressure error on these grids because interface smearing replaces the sharp water--epoxy impedance jump with a mixed layer, which transmits more of the wave both into and out of the cylinder.
This increases the early transmission into the cylinder and weakens the later pressure radiated back into the water.
The deviatoric-energy diagnostic shows the same effect.
HLL stores less elastic energy in the cylinder, while HLLC and anchored HLLD remain closer to the modal reference.
Anchored HLLD lies closest at both resolutions because it additionally resolves the shear waves generated in the solid.
The difference is most visible in $E_{\mathrm{dev}}$ at later times, after the shear waves have developed in the cylinder.

\subsection{Shock interaction with an epoxy half-cylinder}
\label{subsec:s62-shock-cylinder}

This case extends the comparison of HLL, HLLC, and anchored HLLD from the subsonic regime to a supersonic shock.
It follows the shock--cylinder problem of Rodriguez and Johnsen~\cite{RodriguezJohnsen2019} in geometry and Mach number, with different material properties and constitutive law, so that the present anchored HLLD solution can be compared with their published results qualitatively.
A Mach $1.006$ water shock, with a pressure ratio of approximately 119 across the shock, impacts an epoxy half-cylinder.
The domain is $[0, 1.2]\times[0, 0.6]$~m, with the half-cylinder of radius $R_0 = 0.2$~m centered at $(0.6, 0)$ on the $y = 0$ symmetry plane.
The initial condition comprises three regions: quiescent water with $p = 10^5$~Pa ahead of the shock, a post-shock strip at $x \in [0, 0.19]$ ($u = 7.20$~m/s, $p = 1.187\times 10^7$~Pa, $\rho \approx 1004.4$~kg/m$^3$, giving $M \approx 1.006$), and the epoxy half-cylinder.
Resolution is quoted in cells per cylinder radius, and results are shown at $192$ and $384$ cells per radius.
The finer case is a $2304\times 1152$ grid, beyond the $32$ to $256$ cells per radius of Rodriguez and Johnsen, and the case is run to $t = 1500\,\mu$s.

\begin{figure}[htbp]
\centering
\includegraphics[width=\textwidth]{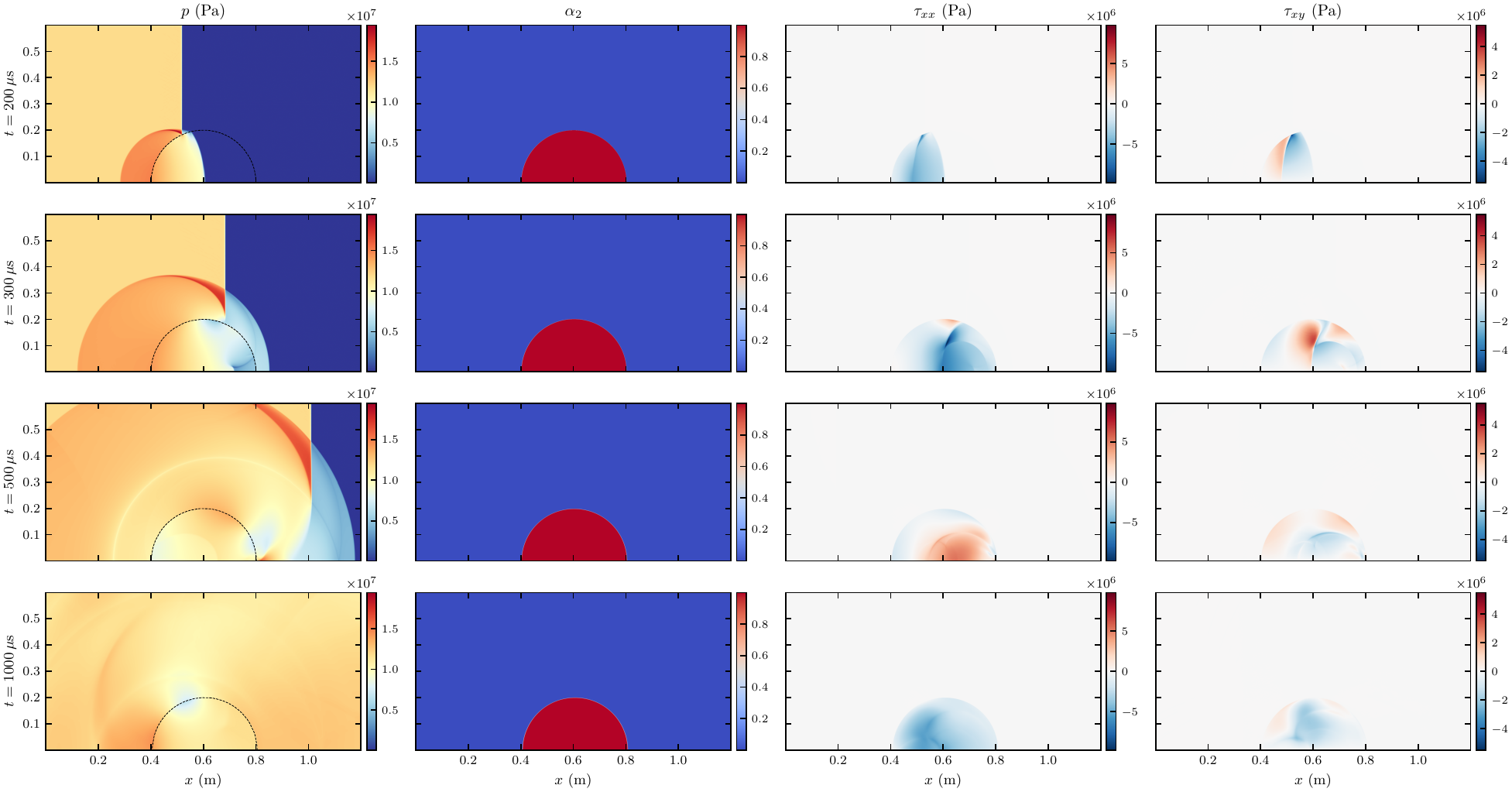}
\caption{Shock--cylinder interaction: $p$, $\alpha_2$, $\tau_{xx}$, $\tau_{xy}$ at
$t = 200$, $300$, $500$, $1000\,\mu$s (HLLD, $384$ cells per radius, $2304\times 1152$).
The dashed contour in the $p$ column marks the instantaneous $\alpha_2 = 0.5$ interface.}
\label{fig:s62-contour}
\end{figure}

As in the acoustic case, the stresses stay cleanly contained within the epoxy half-cylinder (Figure~\ref{fig:s62-contour}).
Unlike the linear acoustic problem, the finite-amplitude shock displaces the inclusion slightly downstream, visible in the $\alpha_2$ column, while the interface itself remains clean.

\begin{figure}[htbp]
\centering
\includegraphics[width=\textwidth]{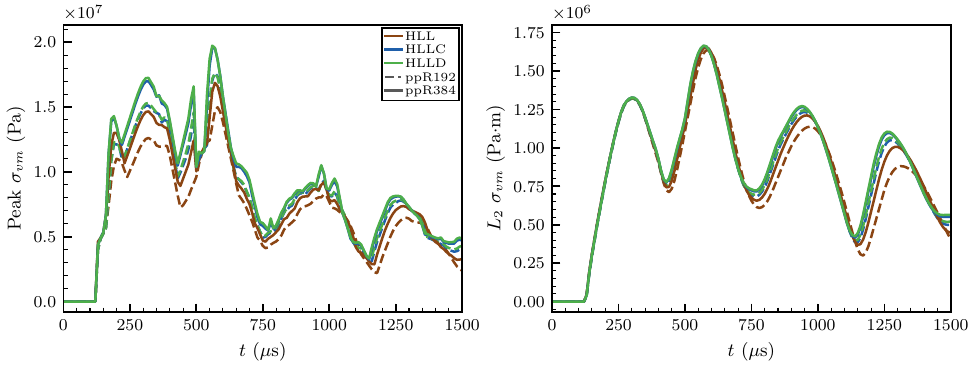}
\caption{Peak (left) and $L_2$ (right) von Mises stress $\sigma_{vm}(t)$, with $\sigma_{vm} = \sqrt{\tau_{xx}^2 - \tau_{xx}\tau_{yy} + \tau_{yy}^2 + 3\tau_{xy}^2}$, for the shock--cylinder case at $192$ (dashed) and $384$ (solid) cells per cylinder radius.}
\label{fig:s62-histories}
\end{figure}

Two measures of the von Mises stress track the resolution dependence (Figure~\ref{fig:s62-histories}).
The three solvers agree at early times and separate later, with HLL reaching lower maxima than HLLC and anchored HLLD at both resolutions.
HLLC and anchored HLLD remain close, and anchored HLLD retains the highest peak stress.
The pointwise peak $\sigma_{vm}$ has not converged at the resolutions shown for any solver.
Pointwise stress maxima are a stringent convergence diagnostic in finite-volume schemes: capturing schemes retain $O(1)$ pointwise errors near shocks even at formally higher order~\cite{EngquistSjogreen1998} and converge sublinearly at the linearly degenerate contacts that mark material interfaces~\cite{BanksAslamRider2008}, so pointwise maxima can converge more slowly than integrated norms, consistent with the published Rodriguez and Johnsen HLL results for the analogous configuration.
The integrated $L_2\,\sigma_{vm}$, by contrast, changes far less between the two resolutions for all three solvers.

\subsection{Three-dimensional shock interaction with a tilted ellipsoid}
\label{subsec:s63-3d-ellipsoid}

A Mach $1.006$ water shock interacts with a tilted epoxy ellipsoid in three dimensions.
The domain is $[0, 1.5]\times[0, 1]\times[0, 1]$~m and the ellipsoid has semi-axes $(0.25, 0.12, 0.10)$~m centered at $(0.8, 0.5, 0.5)$, rotated $-30^\circ$ about the $y$-axis and a further $-10^\circ$ about the $z$-axis, which breaks the spatial symmetries.
The shock is initialized left of the ellipsoid at $x = 0.2$, the production grid is $1200\times 800\times 800$, and the case is run to $t = 500\,\mu$s.
The three-dimensional von Mises stress is \(\sigma_{vm}=\sqrt{\tfrac12\bigl[(\tau_{xx}-\tau_{yy})^2+(\tau_{yy}-\tau_{zz})^2+(\tau_{zz}-\tau_{xx})^2\bigr]+3\bigl(\tau_{xy}^2+\tau_{xz}^2+\tau_{yz}^2\bigr)}\).

\begin{figure}[htbp]
\centering
\includegraphics[width=\textwidth]{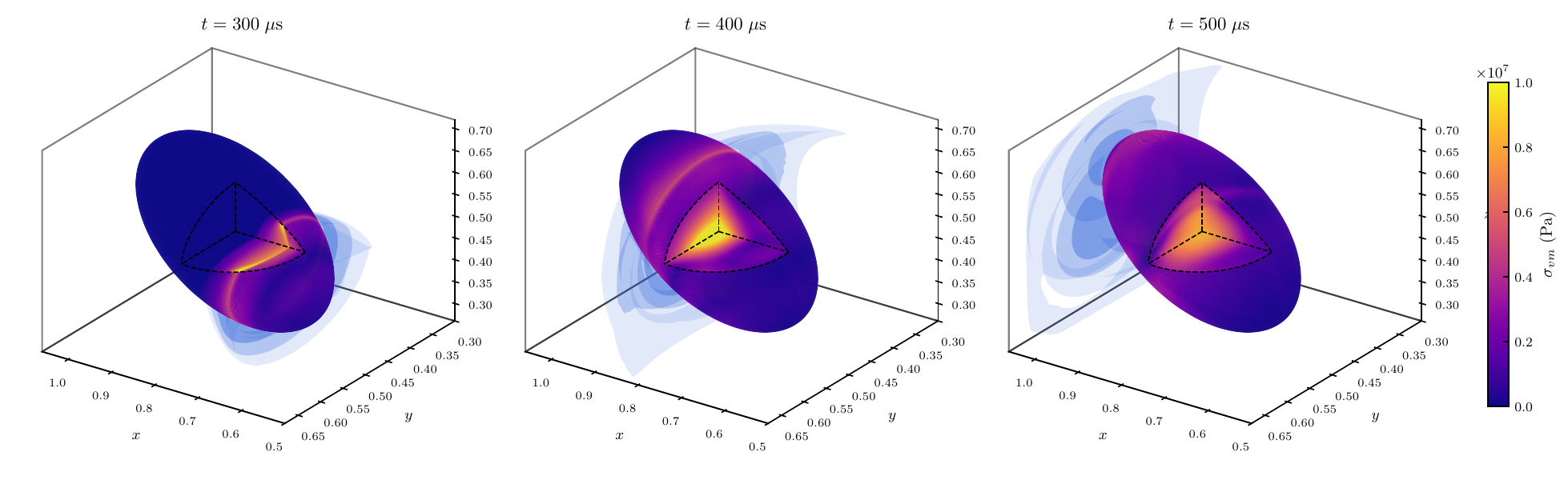}
\caption{The $\alpha_2 = 0.5$ material interface at three successive times, with
one octant cut away for visualization only (the full ellipsoid is simulated).
The von Mises stress is represented by $\sigma_{vm}$.
Translucent blue shells are iso-levels of the scattered pressure, $p$ minus its
planar $y,z$-average at each $x$ (the incident planar wave).}
\label{fig:3d-evolution}
\end{figure}

\begin{figure}[htbp]
\centering
\includegraphics[width=\textwidth]{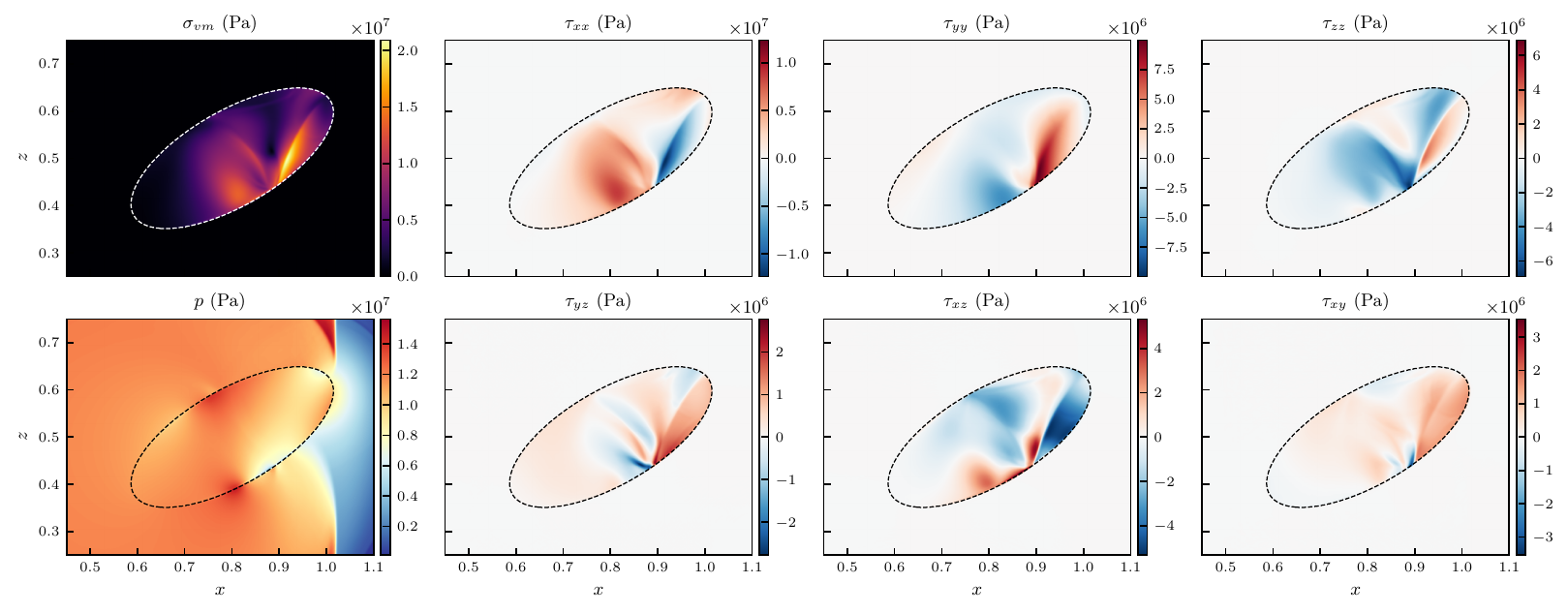}
\caption{Stress-tensor decomposition on the $y = 0.5$ midplane at $t = 500\,\mu$s. Dashed contours mark the
$\alpha_2 = 0.5$ material interface.}
\label{fig:3d-stress}
\end{figure}

The shock passes over the ellipsoid and excites all six stress components, each cleanly contained within the solid (Figures~\ref{fig:3d-evolution} and~\ref{fig:3d-stress}).
In particular, the three shear components $\tau_{xy}$, $\tau_{xz}$, and $\tau_{yz}$ stay confined to the solid across the interface, consistent with the shear-resolving capability of the three-dimensional anchored HLLD formulation detailed in~\ref{app:3d}.

\FloatBarrier
\section{Conclusions}
\label{sec:conclusions}

We have developed an anchored HLLD Riemann solver for the multicomponent hypoelastic diffuse-interface model of compressible fluid--solid interaction, whose deviatoric stress evolves by an objective hypoelastic rate law, outside conservation-law form.
The difficulty specific to this model is that the nonconservative shear coefficient $\rho(G+\tau_{xx})$ also sets the impedance of the resolved shear waves, part of the eigenstructure the solver must capture (Section~\ref{sec:governing}).
The Kapila-type treatment that pairs such a coefficient with an interface-consistent velocity trace from the Riemann fan yields the HLLC member, but it acts outside the solve and does not reach the shear pair.
Anchored HLLD carries the coefficient inside the solve: the same reconstructed left and right states are passed to the HLLD solver twice, anchored in turn to each adjacent cell-centered state, and every cell advances with the flux anchored to its own state (Section~\ref{sec:method}).
To our knowledge, no previous HLLD solver has been formulated for a multicomponent hypoelastic diffuse-interface fluid--solid model with no separately computed interface.
The HLL two-wave baseline of Rodriguez and Johnsen~\cite{RodriguezJohnsen2019}, the HLLC adaptation, and anchored HLLD form a hierarchy of increasing wave resolution within one framework, implemented in the open-source exascale solver MFC~\cite{MFC2021,MFC2026}.

These constructions rest on local structural properties (\ref{sec:analysis}).
The central one shows that the dual pass is not merely a convenience: an irreducibility criterion (Proposition~\ref{prop:anchoring-criterion}) gives a sufficient condition under which an anchored face flux admits no single-pass interface-consistent form, met on the HLLD inner shear pair whenever the shear jump is active, and absent in the fluid limit, where the dual pass is no longer required and the construction reduces to a single pass.
On smooth data the HLLC and anchored HLLD members are consistent at their reconstruction order (Propositions~\ref{prop:smooth-hllc} and~\ref{prop:smooth-hlld}).
Under the material-contact hypotheses of Proposition~\ref{prop:contact}, one first-order forward Euler step preserves the normal velocity and the tangential fields exactly, and the total pressure $p_T = p - \tau_{xx}$ up to a quantified elastic-energy blend residual when the density and normal stress do not both jump.
It conserves the species masses exactly, and within each directional sweep the normal momentum, globally under periodic or isolated boundaries (Proposition~\ref{prop:pairwise}).
The tangential momentum and total energy are conserved quantities of the hypoelastic model that are not conserved by the scheme, and their drift vanishes under refinement at the reconstruction's face-jump order on smooth fields (Proposition~\ref{prop:drift-rate}).

The one-dimensional tests (Section~\ref{sec:results-1d}) confirm several of the analytical results.
All three solvers recover their formal order on smooth eigenmodes, including under rotation and within a smoothly varying mixture.
The inactive-shear cases reproduce the reduction of anchored HLLD to HLLC.
With the shear waves active, anchored HLLD resolves the contact and shear pair together, HLLC the contact alone, and HLL neither sharply.
Anchored HLLD attains the lowest error on most of the shear-bearing problems tested here.
The conservation drift in tangential momentum and total energy converges at least at the reconstruction's face-jump order on smooth fields, consistent with the analysis, and at first order in the presence of material interfaces and shocks.

At production scale in two dimensions (Section~\ref{sec:results-2d}), the shear- and contact-resolution advantage persists.
In transient acoustic scattering by an epoxy cylinder, HLLC and anchored HLLD preserve the material contact almost exactly while HLL diffuses it.
The solutions converge toward the semi-analytic reference in pressure and stored elastic energy.
Anchored HLLD lies closest to the reference in the stored elastic energy.
Under shock loading of an epoxy half-cylinder, anchored HLLD again gives the sharpest stress response, HLLC follows closely, and HLL gives the most diffused response.
A three-dimensional shock over a tilted epoxy ellipsoid, which excites all six stress components (Section~\ref{subsec:s63-3d-ellipsoid}), and a complex-inclusion configuration with sharp corners and narrow gaps (\ref{app:complex}) test robustness at production scale.
In each case the anchored HLLD fields show clean material interfaces and deviatoric stresses confined to the solid.

The analytical results are local and structural: we do not claim path-conservative convergence in the Dal Maso--LeFloch--Murat sense~\cite{DalMasoLeFlochMurat1995}.
For future work, two extensions follow from the construction.
First, the underlying principle, placing a nonconservative coefficient inside the Riemann solve when a resolved wave depends on it, is not specific to hypoelasticity.
The irreducibility criterion of Proposition~\ref{prop:anchoring-criterion} is stated for a general nonconservative system, and the construction may apply to other hyperbolic systems where a nonconservative coefficient enters a resolved inner wave.
Second, because the scheme applies the same operations to mixture and pure-component cells, it is a candidate for genuinely mixed media such as saturated porous solids, and resolves the shear waves in the mixture and at its boundaries alike.

\clearpage
\appendix
\setcounter{figure}{0}%
\renewcommand{\theproposition}{\Alph{section}.\arabic{proposition}}
\renewcommand{\theremark}{\Alph{section}.\arabic{remark}}
\section{Analytical properties}
\label{sec:analysis}

Having constructed the schemes in Section~\ref{sec:method} from their Rankine--Hugoniot conditions, we now establish several further properties they possess.
These are smooth consistency, a criterion for when the dual pass is required, material-contact equilibrium, pairwise conservation, and a bounded conservation drift.
Recall that a hat denotes a quantity at the anchor cell and a superscript \(\widehat D\) denotes dependence on the anchor side.

\subsection{Smooth consistency at order \texorpdfstring{$r$}{r}}
\label{subsec:smooth-consistency}

For smooth exact solutions, both the HLL/HLLC scheme of Section~\ref{subsec:hllc} and the anchored HLLD member of Section~\ref{subsec:hlld} produce a numerical cell-residual that converges, at the order of the spatial reconstruction, to the exact face-difference residual defined in Eq.~\eqref{eq:smooth-L} below.
The face-difference target agrees with the PDE cell-average residual to second order in the nonconservative term, so consistency beyond second order is stated with respect to~\eqref{eq:smooth-L}.
We establish the result in two propositions, one per scheme.
The HLL/HLLC case is treated first, and the HLLD case follows in Proposition~\ref{prop:smooth-hlld} below.

Let $h := \Delta x$ and let $W: x \mapsto W(x)$ be a smooth solution of~\eqref{eq:quasi1d}.
Choose a cell-average quadrature rule of order at least $r$ that includes the cell center with a fixed nonzero weight.
Reconstruct the conservative values at the other quadrature nodes to order $r$ from the conserved cell averages.
Solve the quadrature relation for the center value and convert it to primitive variables, giving $W_i=W(x_i)+\mathcal O(h^r)$.
This cell-centered state serves as the \emph{anchor} $\widehat W_i$ in the HLLD dual-pass evaluation~\eqref{eq:anchored-update}.
The reconstruction satisfies
\begin{equation}
\label{eq:smooth-recon}
 W^{L,R}_{j+1/2} = W(x_{j+1/2}) + h^r\, A_\pm(x_{j+1/2}) + \mathcal{O}(h^{r+1}),
\end{equation}
with $A_\pm$ Lipschitz in $x$.
Let $\mathcal{A}$ be a compact admissible neighborhood of the local smooth state path $\{W(x) : x \in [x_{i-1/2}, x_{i+1/2}]\}$.
We assume that each cell's two reconstructed face pairs use the same nondegenerate algebraic branches, including fixed choices in the wave-speed min/max formulas, with all denominators bounded away from zero.
The fixed-branch formulas have $C^{1,1}$ bounds uniform under refinement and, for HLLD, over the admissible anchors.

Both schemes target the same \emph{exact face-difference residual}
\begin{multline}
\label{eq:smooth-L}
 \mathcal{L}_i = -\frac{F(W(x_{i+1/2})) - F(W(x_{i-1/2}))}{h} \\
 + N^u_i\, \frac{u(W(x_{i+1/2})) - u(W(x_{i-1/2}))}{h}
 + N^v_i\, \frac{v(W(x_{i+1/2})) - v(W(x_{i-1/2}))}{h},
\end{multline}
where $N^u(W), N^v(W)$ are the nonconservative-coefficient column vectors of the system~\eqref{eq:quasi1d}, and $N^u_i := N^u(W_i),\, N^v_i := N^v(W_i)$.

\begin{proposition}[HLL/HLLC scheme, smooth consistency at order $r$]
\label{prop:smooth-hllc}
The semi-discrete numerical residual of the HLL/HLLC scheme with HLLC transport traces is
\begin{equation}
\label{eq:smooth-bb-Lh}
 \mathcal{L}^h_i = -\frac{\mathcal{F}^c_{i+1/2} - \mathcal{F}^c_{i-1/2}}{h}
 + N^u_i\, \frac{\Psi_{u,i+1/2} - \Psi_{u,i-1/2}}{h}
 + N^v_i\, \frac{\Psi_{v,i+1/2} - \Psi_{v,i-1/2}}{h}.
\end{equation}
Under these assumptions,
\begin{equation}
\label{eq:smooth-bb-claim}
 \mathcal{L}^h_i = \mathcal{L}_i + \mathcal{O}(h^r).
\end{equation}
\end{proposition}

\begin{proof}
The argument has two ingredients: an equal-state identity for the face quantities, and the standard finite-volume smooth-perturbation mechanism that converts equal-state consistency into cell-residual consistency at order~$r$.

We first establish the equal-state identity.
Collect the face quantities and their smooth-state counterparts into
\begin{equation*}
 \mathcal{G}(W_L, W_R) := \begin{pmatrix}\mathcal{F}^c(W_L, W_R) \\ \Psi_u(W_L, W_R) \\ \Psi_v(W_L, W_R)\end{pmatrix},
 \qquad
 g(W) := \begin{pmatrix}F(W) \\ u(W) \\ v(W)\end{pmatrix}.
\end{equation*}
At $W_L = W_R = W$, $\mathcal{G}(W, W) = g(W)$.
The HLL traces~\eqref{eq:Psi-hll} reduce to $u(W),\, v(W)$ by inspection.
The HLLC star-branch traces~\eqref{eq:Psi-hllc} reduce through $S_M = u(W),\, \xi_K = 1,\, v^{*} = v(W)$ from Eqs.~\eqref{eq:SM-hllc}--\eqref{eq:vstar-hllc}, and the full nine-row conservative flux collapses to $F(W)$.

We then convert this to cell-residual consistency.
$\mathcal{G}$ is $C^{1,1}$ in a neighborhood of the diagonal set $\{(W(x), W(x)) : x \in [x_{i-1/2}, x_{i+1/2}]\}$ (smooth on the fixed branch, with bounded denominators).
Taylor expansion at the smooth-state diagonal, using the equal-state identity established above and the reconstruction~\eqref{eq:smooth-recon}, gives
\begin{equation}
\label{eq:smooth-bb-face-error}
 \mathcal{G}(W^L_{j+1/2}, W^R_{j+1/2})
 = g(W(x_{j+1/2})) + h^r\, B(x_{j+1/2}) + \mathcal{O}(h^{r+1}),
\end{equation}
where \[ B(x) := D_L\mathcal{G}(W(x), W(x))\, A_-(x) + D_R\mathcal{G}(W(x), W(x))\, A_+(x).
\] $B$ is Lipschitz in $x$ on $\mathcal{A}$ (composition of $C^{1,1}$ derivatives with Lipschitz $A_\pm$), so $B(x_{i+1/2}) - B(x_{i-1/2}) = \mathcal{O}(h)$ and the cell-$i$ divided difference of~\eqref{eq:smooth-bb-face-error} reads
\begin{equation}
\label{eq:smooth-bb-divided-diff}
 \frac{\mathcal{G}(W^L_{i+1/2}, W^R_{i+1/2}) - \mathcal{G}(W^L_{i-1/2}, W^R_{i-1/2})}{h}
 = \frac{g(W(x_{i+1/2})) - g(W(x_{i-1/2}))}{h} + \mathcal{O}(h^r).
\end{equation}
Reading~\eqref{eq:smooth-bb-divided-diff} component-by-component ($\mathcal{F}^c, \Psi_u, \Psi_v$) and substituting into~\eqref{eq:smooth-bb-Lh} with $N^u_i, N^v_i$ bounded matches~\eqref{eq:smooth-L} up to $\mathcal{O}(h^r)$, giving~\eqref{eq:smooth-bb-claim}.
\end{proof}

\begin{proposition}[Anchored HLLD, smooth consistency at order $r$]
\label{prop:smooth-hlld}
The semi-discrete numerical residual of the anchored HLLD member of Section~\ref{subsec:hlld} is
\begin{equation}
\label{eq:smooth-hlld-Lh}
 \mathcal{L}^h_i = -\frac{
 \mathcal{F}_{\mathrm{HLLD}}(W^L_{i+1/2}, W^R_{i+1/2};\, \widehat W_i)
 -
 \mathcal{F}_{\mathrm{HLLD}}(W^L_{i-1/2}, W^R_{i-1/2};\, \widehat W_i)
 }{h},
\end{equation}
where $\mathcal{F}_{\mathrm{HLLD}}(W_L,W_R;\widehat W)$ is the anchored face flux of Eq.~\eqref{eq:hlld-flux} with the anchor written as an explicit argument, and the dual-pass update~\eqref{eq:anchored-update} uses cell $i$'s anchor $\widehat W_i$ on both of its adjacent faces.
Under these assumptions, with $\widehat C^{\widehat D}$ from~\eqref{eq:Chat} bounded below on $\mathcal A$,
\begin{equation}
\label{eq:smooth-hlld-claim}
 \mathcal{L}^h_i = \mathcal{L}_i + \mathcal{O}(h^r).
\end{equation}
\end{proposition}

\begin{proof}
HLLD absorbs the nonconservative contribution into the single anchored augmented flux $F^{\widehat D}(W) = F(W) - N^u(\widehat W^{\widehat D})\,u(W) - N^v(\widehat W^{\widehat D})\,v(W)$ of~\eqref{eq:anchored-flux}.
Below we write $F(W;\widehat W):=F^{\widehat D}(W)$ to show the anchor explicitly.
The proof has three ingredients: an equal-state identity returning $F^{\widehat D}(W)$ (HLLD-specific); the Lipschitz-$B$ adjacent-face cancellation in the cell-residual-consistency argument of Proposition~\ref{prop:smooth-hllc}, invoked with the anchor as a fixed parameter; and an algebraic recovery converting the augmented-flux face-difference to $\mathcal L_i$.

We first establish the equal-state identity.
At $W_L = W_R = W$ and any admissible anchor $\widehat W^{\widehat D}$, every sampled branch of~\eqref{eq:hlld-flux} returns the augmented flux:
\begin{equation}
\label{eq:smooth-hlld-equalstate}
 \mathcal{F}_{\mathrm{HLLD}}(W, W;\, \widehat W^{\widehat D}) = F^{\widehat D}(W).
\end{equation}
The HLLC outer-contact relations~\eqref{eq:SM-hllc}--\eqref{eq:Tstar-hllc} give $S_M = u(W),\, p_T^* = p_T(W),\, \xi_K = 1$.
The outer-star increments~\eqref{eq:hlld-alpha}--\eqref{eq:hlld-tauyy} and the energy star~\eqref{eq:hlld-energy-star} carry the prefactor $(S_M - u_K) = 0$ and vanish.
The outer tangential~\eqref{eq:hlld-outer-tangential} copies the physical values.
The inner double-star pair~\eqref{eq:vstarstar}--\eqref{eq:tauxystarstar} and the corresponding $\tau_{yy}$ and energy increments~\eqref{eq:tauyy-Lstarstar}--\eqref{eq:energy-Rstarstar} reduce by the same mechanism applied to $(v_R - v_L) = (\tau_{xy,R} - \tau_{xy,L}) = 0$.

We then carry out the adjacent-face cancellation at fixed anchor.
With $\widehat W_i$ held fixed, the fixed-branch formula for $\mathcal F_{\mathrm{HLLD}}(\cdot,\cdot;\widehat W_i)$ is $C^{1,1}$ near the smooth-state diagonal (HLL/HLLC denominators and $\widehat C^{\widehat D}$ bounded away from zero on $\mathcal A$).
The cell-residual-consistency argument of Proposition~\ref{prop:smooth-hllc} applies with $\mathcal{G} \mapsto \mathcal{F}_{\mathrm{HLLD}}(\cdot, \cdot;\, \widehat W_i)$ and $g \mapsto F(\cdot\,;\widehat W_i)$.
The anchor is a fixed parameter and no anchor-derivative enters.
This yields the intermediate identity
\begin{equation}
\label{eq:smooth-hlld-intermediate}
 \mathcal{L}^h_i = -\frac{
 F(W(x_{i+1/2});\widehat W_i) - F(W(x_{i-1/2});\widehat W_i)
 }{h}
 + \mathcal{O}(h^r).
\end{equation}

Finally, we recover the residual~\eqref{eq:smooth-L}.
The dual-pass design uses the same anchor $\widehat W_i$ at both adjacent faces of cell $i$.
Linearity of the anchor-dependent part of $F(\cdot\,;\widehat W_i)$ in $(u, v)$ at fixed anchor then lets $N^u(\widehat W_i), N^v(\widehat W_i)$ factor out of the adjacent-face difference:
\begin{equation}
\label{eq:smooth-hlld-recovery}
\begin{aligned}
 F(W(x_{i+1/2});\widehat W_i) - F(W(x_{i-1/2});\widehat W_i)
 &= F(W(x_{i+1/2})) - F(W(x_{i-1/2})) \\
 &\quad - N^u_i\,\bigl[u(W(x_{i+1/2})) - u(W(x_{i-1/2}))\bigr] \\
 &\quad - N^v_i\,\bigl[v(W(x_{i+1/2})) - v(W(x_{i-1/2}))\bigr],
\end{aligned}
\end{equation}
using $\widehat W_i = W_i$.
Substituting into~\eqref{eq:smooth-hlld-intermediate} and comparing to~\eqref{eq:smooth-L} gives the claim.
\end{proof}

\begin{remark}[Same face-difference target $\mathcal L_i$ for both schemes]
\label{rem:smooth-shared-target}
Propositions~\ref{prop:smooth-hllc} and~\ref{prop:smooth-hlld} converge at order $r$ to the same face-difference residual $\mathcal L_i$ of~\eqref{eq:smooth-L} despite different discrete decompositions.
The algebraic bridge is the residual-recovery step of Proposition~\ref{prop:smooth-hlld}: had the dual-pass evaluation used face-varying anchors instead of cell $i$'s fixed $\widehat W_i$ at both adjacent faces, the factorization~\eqref{eq:smooth-hlld-recovery} would have acquired spurious $\partial_x N^u\cdot u + \partial_x N^v\cdot v$ contributions, breaking the identification.
\end{remark}

\subsection{Irreducibility criterion}
\label{subsec:anchoring-criterion}

Equation~\eqref{eq:ic-representation} of Section~\ref{subsec:ic-nc-update} defines when an anchored face flux admits an interface-consistent representation.
By Lemma~\ref{lem:dual-pass-to-bb} of Section~\ref{subsec:hlld}, such a representation lets the anchored dual-pass cell update collapse to the single-pass update of Eq.~\eqref{eq:oos-update}.
The proposition below gives a sufficient condition under which a branch fails to admit this form.

\begin{proposition}[Irreducibility criterion]
\label{prop:anchoring-criterion}
Fix reconstructed face data \((W_L,W_R)\) and one sampled branch of the anchored face operator \(\mathcal R(W_L,W_R;\widehat W)\), considered over the set of admissible anchors for which this branch is sampled.
Suppose there exists a fixed row vector \(\ell^T\), independent of the anchor, such that
\begin{equation}
\label{eq:source-null-row}
\ell^T N^u(\widehat W) = 0,
\qquad
\ell^T N^v(\widehat W) = 0
\end{equation}
for every admissible anchor on this branch, while
\begin{equation}
\label{eq:source-null-flux-varies}
\ell^T \mathcal R(W_L,W_R;\widehat W)
\end{equation}
depends nontrivially on \(\widehat W\).
Then this branch admits no representation of the form Eq.~\eqref{eq:ic-representation}.
\end{proposition}

\begin{proof}
Assume, for contradiction, that Eq.~\eqref{eq:ic-representation} holds on this branch with some anchor-independent triple \(\mathcal F^c,\Psi_u,\Psi_v\).
Left-multiplying by~\(\ell^T\) and using~\eqref{eq:source-null-row} to annihilate the two source terms gives \[ \ell^T \mathcal R(W_L,W_R;\widehat W) = \ell^T \mathcal F^c(W_L,W_R), \] which is independent of the anchor.
This contradicts Eq.~\eqref{eq:source-null-flux-varies}.
\end{proof}

\begin{remark}[Scope of the criterion]
\label{rem:scope-of-criterion}
The criterion does not forbid anchor-dependent star states.
Their contribution to the sampled face flux is absorbed into \(-N^u(\widehat W)\Psi_u - N^v(\widehat W)\Psi_v\) with anchor-independent traces whenever such a representation exists.
The criterion applies only when anchor-dependence appears in a row annihilated by both source columns, a direction the out-of-solver source cannot reach.
\end{remark}

\begin{remark}[Application to anchored HLLD]
\label{rem:hlld-shear-irreducible}
We apply Proposition~\ref{prop:anchoring-criterion}, stated for the general nonconservative system of Eq.~\eqref{eq:quasi1d}, to the anchored HLLD construction of Section~\ref{subsec:hlld} on the hypoelastic system of Eq.~\eqref{eq:nc}.
The \(\rho v\) row is conservative-core: both \(N^u(\widehat W)\) and \(N^v(\widehat W)\) have a zero entry there for every anchor, so the choice \(\ell^T = e_{\rho v}^T\) satisfies the source-column hypothesis, \(e_{\rho v}^T N^u(\widehat W) = e_{\rho v}^T N^v(\widehat W) = 0\), identically in \(\widehat W\).

On the nondegenerate branch \(L^{**}\) of Eq.~\eqref{eq:hlld-flux}, sampled when \(S_{L*}^{\widehat D}\le0\le S_M\), the flux is the outer-star flux plus the inner-wave increment \(S_{L*}^{\widehat D}\bigl(U_L^{**,\widehat D}-U_L^{*}\bigr)\).
The outer-star flux is anchor-free on the \(\rho v\) row, so the anchor enters this row only through this increment.
Substituting Eqs.~\eqref{eq:inner-speeds} and~\eqref{eq:vstarstar}, its \(\rho v\) component reads
\begin{equation}
\label{eq:hlld-rhov-flux-active}
S_{L*}^{\widehat D}\,\rho_L^*\bigl(v^{**,\widehat D}-v_L\bigr)
=
\tfrac12\rho_L^* S_M\,\Delta v-\tfrac12\Delta\tau
\;-\;\tfrac12\Delta v\,\widehat C^{\widehat D}
\;+\;\frac{\rho_L^* S_M\,\Delta\tau}{2\,\widehat C^{\widehat D}},
\end{equation}
with \(\Delta v:=v_R-v_L\) and \(\Delta\tau:=\tau_{xy,R}-\tau_{xy,L}\).
The first two terms are anchor-independent.
The two \(\widehat C^{\widehat D}\)-bearing terms have coefficients \(-\tfrac12\Delta v\) and \(\tfrac12\rho_L^* S_M\,\Delta\tau\), so for an active shear jump (\(\Delta v\neq0\), or \(S_M\,\Delta\tau\neq0\)) at least one is nonzero and the \(\rho v\) flux varies with the anchor on the active branch.
A single anchor-independent flux therefore cannot reproduce it.
Proposition~\ref{prop:anchoring-criterion} then forbids the interface-consistent form on \(L^{**}\), and identically on \(R^{**}\).
\end{remark}

\subsection{Material-contact equilibrium}
\label{subsec:material-contact-equilibrium}

\begin{proposition}[Discrete material-contact equilibrium]
\label{prop:contact}
Consider a piecewise-constant material contact at a cell face between states $W_L,W_R$ with positive densities and shear moduli.
Suppose $u$, $p_T=p-\tau_{xx}$, $v$, and $\tau_{xy}$ have common values $u_0,p_{T,0},v_0,\tau_{xy,0}$ on both sides.
Assume common stiffened-gas parameters $\gamma,\pinf$ and $\Delta\rho\,\Delta\tau_{xx}=0$, where $\Delta(\cdot)=(\cdot)_R-(\cdot)_L$.
For $\nu=|u_0|\Delta t/\Delta x\in[0,1]$, one first-order forward Euler step of either anchored HLLD or HLLC preserves $u,v,\tau_{xy}$ exactly.
The total pressure changes by the elastic-energy blend residual in Eq.~\eqref{eq:pT-onestep}, which is locally quadratic in the face jumps for $\rho,G>0$.
\end{proposition}

\begin{proof}
For $u_0>0$, we first show that the cell receiving the contact has the upwind blend
\begin{equation}
\label{eq:contact-blend}
U^{n+1}=(1-\nu)U_R+\nu U_L,
\qquad
\rho^{n+1}=(1-\nu)\rho_R+\nu\rho_L.
\end{equation}
The case $u_0<0$ interchanges $L,R$, and $u_0=0$ leaves both states unchanged.
The contact relations (Eqs.~\eqref{eq:SM-hllc}--\eqref{eq:xi}) give $S_M=u_0$, $p_T^*=p_{T,0}$, and $\xi_K=1$.
Thus the HLLD outer-star states equal their corresponding input states.
Uniform $v$ and $\tau_{xy}$ make the inner-state increments vanish as well.
For HLLC, the same data give unchanged star states and constant face traces $\Psi_u=u_0$, $\Psi_v=v_0$, so the out-of-solver source vanishes.
The HLLC flux is therefore the upwind physical flux.
The HLLD flux adds $-N^u(\widehat W_i)u_0-N^v(\widehat W_i)v_0$ to that flux.
This offset cancels between the two faces of cell $i$.
The common traction and work terms of the physical flux cancel as well, leaving Eq.~\eqref{eq:contact-blend} for both schemes.

The blended variables $\rho u$, $\rho v$, and $\rho\tau_{xy}$ are $\rho^{n+1}$ times $u_0,v_0,\tau_{xy,0}$, respectively.
Dividing by $\rho^{n+1}$ therefore recovers these three values exactly.
The normal stress instead recovers as
\[
\tau_{xx}^{n+1}
=\frac{(1-\nu)\rho_R\tau_{xx,R}+\nu\rho_L\tau_{xx,L}}{\rho^{n+1}}.
\]
Since the velocities are common, the kinetic-energy density blends linearly.
The caloric closure~\eqref{eq:caloric} with common $\gamma,\pinf$ then gives the total-pressure deviation
\begin{equation}
\label{eq:pT-onestep}
p_T^{n+1}-p_{T,0}
=-\nu(1-\nu)\frac{\Delta\rho\,\Delta\tau_{xx}}{\rho^{n+1}}
+(\gamma-1)\bigl[(1-\nu)e_{\mathrm{el},R}+\nu e_{\mathrm{el},L}-e_{\mathrm{el}}(W^{n+1})\bigr],
\end{equation}
where $p_K=p_{T,0}+\tau_{xx,K}$ has been used.

The first term vanishes by $\Delta\rho\,\Delta\tau_{xx}=0$.
In the elastic-energy residual, the linear Taylor terms cancel by Eq.~\eqref{eq:contact-blend}.
For $\rho,G>0$, $e_{\mathrm{el}}$ is a smooth function of the conservative variables, so the remaining term is locally quadratic in the face jumps.
It vanishes whenever $e_{\mathrm{el}}$ is affine along the blend.
For equal densities and a common shear modulus, the total-pressure residual reduces to
\[
\frac{\nu(1-\nu)(\gamma-1)}{4G}
\bigl[(\Delta\tau_{xx})^2+(\Delta\tau_{yy})^2\bigr].\qedhere
\]
\end{proof}

For stress-free fluid data with $G=0$ and $e_{\mathrm{el}}=0$ throughout, the same contact update preserves $p_T$ exactly.

\begin{remark}
\label{rem:contact-traction-invariant}
The relevant normal invariant is the total pressure $p_T = p - \tau_{xx}$, not the pressure $p$.
Solid--solid and fluid--solid contact tests should be diagnosed in $p_T$, since $p$ may vary across the contact whenever $\tau_{xx}$ does.
\end{remark}

\subsection{Pairwise face conservation}
\label{subsec:pairwise-conservation}

The anchored dual pass produces a distinct flux for each anchor (Section~\ref{subsec:hlld}).
The following result identifies the subset of rows of the conservative variables \(U\) on which the two fluxes nevertheless agree at every interior face.

\begin{proposition}[Pairwise face conservation]
\label{prop:pairwise}
At every interior face $x_{i+1/2}$, the anchored HLLD dual-pass face fluxes for the species masses $m_1, m_2$ and the normal momentum $\rho u$ satisfy
\begin{equation}
\label{eq:pairwise}
\mathcal F^{\widehat L}_{i+1/2}[k] = \mathcal F^{\widehat R}_{i+1/2}[k]
\qquad k \in \{m_1, m_2, \rho u\}.
\end{equation}
\end{proposition}

\begin{proof}
On $k \in \{m_1, m_2, \rho u\}$, every quantity entering the sampled flux is anchor-free.
First, the source columns $N^u, N^v$ vanish on these rows by Eq.~\eqref{eq:nc}, so $F^{\widehat D}[k] = F[k]$.
Second, the outer-$K$ star values $m_{j,K}^{*} = m_{j,K}\,\xi_K$ and $(\rho u)_K^{*} = \rho_K^{*}\,S_M$ depend only on the reconstructed face data $(W_L, W_R)$.
Third, across the inner waves, density is frozen (Eq.~\eqref{eq:hlld-rho-dbl}) and $u = S_M$ is fixed by the contact condition (Eq.~\eqref{eq:hlld-contact}), so partial masses and normal momentum are unchanged.
The anchor-dependent inner speeds $S_{K*}^{\widehat D}$ may select different branches for the two anchors, but on these rows $U^{**}=U^{*}$ makes the star and double-star fluxes equal.
Both $\widehat D \in \{\widehat L, \widehat R\}$ therefore deliver the same face value on these rows at every branch of Eq.~\eqref{eq:hlld-flux}.
\end{proof}

\begin{remark}[Global conservation]
\label{rem:global-conservation}
By Proposition~\ref{prop:pairwise}, the interior-face contributions to the cell update~\eqref{eq:anchored-update} telescope to zero on rows $\{m_1, m_2, \rho u\}$.
Under periodic or isolated boundary conditions the boundary terms vanish, so the spatial integrals $\sum_i U_i[k]\,\Delta x$ are globally preserved on these rows.
In two and three dimensions the same argument applies per sweep, so the species masses are conserved exactly in every direction.
The normal momentum of one direction is the tangential momentum of the others, and there it inherits the drift of Proposition~\ref{prop:drift-rate}.
\end{remark}

\subsection{Drift convergence rate on smooth periodic data}
\label{subsec:drift-rate-analysis}

Among the rows of~\eqref{eq:fullstate}, the species masses $m_1, m_2$, the momenta $\rho u, \rho v$, and the total energy $\rho E$ are conservative at the governing-equation level (their entries in $N^u, N^v$ of~\eqref{eq:nc} vanish), while the volume fraction $\alpha_1$ and the stress components $\rho\tau_{xx}, \rho\tau_{yy}, \rho\tau_{xy}$ are nonconservative at the governing-equation level.
The single-pass HLL and HLLC schemes apply one conservative flux $\mathcal{F}^c$ per face and conserve all five conservative rows exactly.
For anchored HLLD, Proposition~\ref{prop:pairwise} shows that the dual-pass update preserves $m_1, m_2, \rho u$ pairwise at every interior face, hence exactly in spatial integrals under periodic or isolated boundary conditions.
The remaining two conservative rows, tangential momentum $\rho v$ and total energy $\rho E$, are not pairwise-conserved by the anchored construction and acquire an unintended discrete drift.
The following proposition bounds it on smooth periodic data.

\begingroup
\begin{proposition}[Drift rate on tangential momentum and energy]
\label{prop:drift-rate}
Consider the $x$-directed anchored HLLD update on a fixed periodic domain with uniform mesh spacing $h$ and a fixed explicit Runge--Kutta method.
Suppose the stage states remain uniformly smooth over $[0,T]$ in a fixed compact admissible set on which $\widehat C^{\widehat D}>0$.
Assume the reconstruction is at least first-order accurate on smooth data.
Assume the five-wave fan is valid at every face and stage.
For some $q\ge1$, suppose the reconstructed tangential jumps satisfy
\begin{equation}
\label{eq:facejump-q}
v_R-v_L=\mathcal O(h^q),
\qquad
\tau_{xy,R}-\tau_{xy,L}=\mathcal O(h^q),
\end{equation}
uniformly over all faces and stages as $h\to0$.
Then the accumulated drift at the fixed final time $t_{n_T}=T$ satisfies
\begin{equation}
\label{eq:drift-rate}
\sum_i\bigl(U_i^{n_T}[k]-U_i^0[k]\bigr)h=\mathcal O(h^q),
\qquad k\in\{\rho v,\rho E\}.
\end{equation}
\end{proposition}

\begin{proof}
Fix $k\in\{\rho v,\rho E\}$.
Summing one forward Euler stage over the periodic mesh gives
\begin{equation}
\label{eq:same-face-identity}
\sum_i\bigl(U_i^{n+1}[k]-U_i^n[k]\bigr)h
=-\Delta t\sum_f\bigl(\mathcal F_f^{\widehat L}[k]-\mathcal F_f^{\widehat R}[k]\bigr).
\end{equation}
It suffices to show that each face-flux difference is $\mathcal O(h^{q+1})$.
The double-star flux is
\[
\mathcal F_K^{**,\widehat D}[k]
=\mathcal F_K^*[k]+S_{K*}^{\widehat D}
\bigl(U_K^{**,\widehat D}[k]-U_K^*[k]\bigr).
\]
On these rows the source columns vanish and the outer-star quantities depend only on the shared reconstructed data.
Thus $U_K^*[k]$ and $\mathcal F_K^*[k]$ are anchor-independent.
The difference between the two anchored double-star fluxes is
\[
\mathcal F_K^{**,\widehat L}[k]-\mathcal F_K^{**,\widehat R}[k]
=\bigl(S_{K*}^{\widehat L}-S_{K*}^{\widehat R}\bigr)
\bigl(U_K^{**,\widehat L}[k]-U_K^*[k]\bigr)
+S_{K*}^{\widehat R}
\bigl(U_K^{**,\widehat L}[k]-U_K^{**,\widehat R}[k]\bigr).
\]
We bound the speed and state differences in these two terms.

For the speed factor, smooth reconstruction gives $\rho_K^*=\rho_K+\mathcal O(h)$, so the common $\rho_K^*$ is bounded above and away from zero.
Smoothness gives $\widehat C^{\widehat L}-\widehat C^{\widehat R}=\mathcal O(h)$.
The inner speeds $S_{K*}^{\widehat D}=S_M\mp\widehat C^{\widehat D}/\rho_K^*$ are therefore bounded and differ by $\mathcal O(h)$ between anchors.

For the first state difference, the inner-state definitions give $v^{**,\widehat D}-v_K=\mathcal O(h^q)$ and $\tau_{xy}^{**,\widehat D}-\tau_{xy,K}=\mathcal O(h^q)$.
For the conservative states, this gives
\begin{align}
(\rho v)_K^{**,\widehat D}-(\rho v)_K^*
&=\rho_K^*\bigl(v^{**,\widehat D}-v_K\bigr)=\mathcal O(h^q),\notag\\
(\rho E)_K^{**,\widehat D}-(\rho E)_K^*
&=\pm\frac{\rho_K^*}{\widehat C^{\widehat D}}
\Bigl[v^{**,\widehat D}\bigl(\tau_{xy}^{**,\widehat D}-\tau_{xy,K}\bigr)+\tau_{xy,K}\bigl(v^{**,\widehat D}-v_K\bigr)\Bigr]=\mathcal O(h^q),
\label{eq:shear-work-sidesub}
\end{align}
with $+$ for $K=L$ and $-$ for $K=R$.

For the difference between anchors, positivity on the compact admissible set gives
\[
\frac1{\widehat C^{\widehat L}}-\frac1{\widehat C^{\widehat R}}
=\frac{\widehat C^{\widehat R}-\widehat C^{\widehat L}}
{\widehat C^{\widehat L}\widehat C^{\widehat R}}
=\mathcal O(h).
\]
The inner-state formulas then yield
\[
\begin{aligned}
v^{**,\widehat L}-v^{**,\widehat R}
&=\frac{\tau_{xy,R}-\tau_{xy,L}}{2}
\left(\frac1{\widehat C^{\widehat L}}-\frac1{\widehat C^{\widehat R}}\right)
=\mathcal O(h^{q+1}),\\
\tau_{xy}^{**,\widehat L}-\tau_{xy}^{**,\widehat R}
&=\frac{v_R-v_L}{2}
\left(\widehat C^{\widehat L}-\widehat C^{\widehat R}\right)
=\mathcal O(h^{q+1}).
\end{aligned}
\]
Subtracting the two anchored versions of~\eqref{eq:shear-work-sidesub} therefore gives
\[
U_K^{**,\widehat L}[k]-U_K^{**,\widehat R}[k]=\mathcal O(h^{q+1}).
\]

The two products in the flux-difference identity are therefore
\[
\mathcal O(h)\,\mathcal O(h^q)+\mathcal O(1)\,\mathcal O(h^{q+1})
=\mathcal O(h^{q+1}).
\]
The physical and outer-star fluxes are already shared by the anchors.

If the anchors select different branches, the corresponding inner speeds satisfy $S_{K*}^{\widehat L}\leq0\leq S_{K*}^{\widehat R}$, or the reverse ordering.
Thus $\lvert S_{K*}^{\widehat D}\rvert\leq\lvert S_{K*}^{\widehat L}-S_{K*}^{\widehat R}\rvert=\mathcal O(h)$ for either anchor.
The resulting flux difference is therefore $\mathcal O(h)\,\mathcal O(h^q)=\mathcal O(h^{q+1})$.

Summing over $\mathcal O(h^{-1})$ faces gives $\Delta t\,\mathcal O(h^q)$ per stage in~\eqref{eq:same-face-identity}.
The fixed Runge--Kutta weights and $\sum_n\Delta t_n=T$ then give~\eqref{eq:drift-rate}.
\end{proof}

On smooth data, the face-jump order is $q=1$ for first-order upwind reconstruction and $q=3$ for centered MUSCL without a slope limiter.
The impedance lower bound excludes the global fluid limit, in which $\widehat C^{\widehat D}=0$ at every face.
In that limit the update reduces to a single-pass fluid HLLC scheme and produces no conservation drift.
The corresponding $y$-sweep result follows by index permutation.
\par
\endgroup

\begin{remark}[Empirical confirmation]
\label{rem:drift-empirical}
Section~\ref{subsubsec:drift-smooth} reports drift-convergence tests under centered MUSCL with van Leer limiter on a smooth periodic initial condition.
The limiter lowers the face-jump order to $q = 2$ at smooth extrema, so the proposition guarantees $\mathcal O(h^2)$.
Both rows $\rho v$ and $\rho E$ converge at about $\mathcal O(h^3)$.
The observed third-order drift convergence may reflect cancellation between opposite-signed face contributions near smooth extrema.
The same tests under first-order upwind reconstruction, of face-jump order $q = 1$, follow the $\mathcal O(h)$ rate.
Both rates are consistent with the proposition.
The contact and shock cases in Section~\ref{subsubsec:drift-discont} demonstrate first-order drift convergence for both tangential momentum and total energy.
\end{remark}

\clearpage
\section{Three-dimensional system}\label{app:3d}

This appendix gives the three-dimensional extension of the two-dimensional system~\eqref{eq:full2d}.
The \(x\)-directed subsystem is written out explicitly, and the \(y\)- and \(z\)-directed subsystems follow by the index permutations listed at the end.

\subsection{State and conservative flux}

The three-dimensional conservative variables have 13 components.
In the \(x\)-directed subsystem,
\begin{equation}
\label{eq:3d-state}
U=
\bigl[
 \alpha_1,\;m_1,\;m_2,\;\rho u,\;\rho v,\;\rho w,\;\rho E,\;
 \rho\tau_{xx},\;\rho\tau_{yy},\;\rho\tau_{zz},\;
 \rho\tau_{xy},\;\rho\tau_{xz},\;\rho\tau_{yz}
\bigr]^T,
\end{equation}
where the total energy is \(E=e+\tfrac12(u^2+v^2+w^2)\), the internal energy carries the thermal and elastic contributions as in Eq.~\eqref{eq:caloric}, and the elastic energy per unit volume now involves all six independent stress components, \[ e_{\mathrm{el}} = \frac{\tau_{xx}^2+\tau_{yy}^2+\tau_{zz}^2 +2\tau_{xy}^2+2\tau_{xz}^2+2\tau_{yz}^2} {4G}.
\] The conservative flux is
\begin{equation}
\label{eq:3d-flux}
F(W)=
\begin{bmatrix}
\alpha_1 u \\
m_1 u \\
m_2 u \\
\rho u^2 + p - \tau_{xx} \\
\rho uv - \tau_{xy} \\
\rho uw - \tau_{xz} \\
(\rho E+p)\,u - u\tau_{xx} - v\tau_{xy} - w\tau_{xz} \\
\rho u\,\tau_{xx} \\
\rho u\,\tau_{yy} \\
\rho u\,\tau_{zz} \\
\rho u\,\tau_{xy} \\
\rho u\,\tau_{xz} \\
\rho u\,\tau_{yz}
\end{bmatrix}.
\end{equation}

\subsection{Conservative-form nonconservative coefficient columns}

In conservative form parallel to Eq.~\eqref{eq:nc} of Section~\ref{sec:governing}, the three-dimensional \(x\)-directed system reads \[ \partial_t U + \partial_x F(W) = N^u(W)\,\partial_x u + N^v(W)\,\partial_x v + N^w(W)\,\partial_x w, \] with a third column \(N^w\) carrying the \(\partial_x w\) couplings absent in two dimensions.

In the row ordering of Eq.~\eqref{eq:3d-state}, the columns are
\begin{equation}
\label{eq:3d-nc}
\begin{bmatrix} N^u & N^v & N^w \end{bmatrix}(W)=
\begin{bmatrix}
\alpha_1+\Kc & 0 & 0 \\
0 & 0 & 0 \\
0 & 0 & 0 \\
0 & 0 & 0 \\
0 & 0 & 0 \\
0 & 0 & 0 \\
0 & 0 & 0 \\
\rho\bigl(\tfrac43 G+\tau_{xx}\bigr) & 0 & 0 \\
-\rho\bigl(\tfrac23 G+\tau_{yy}\bigr) & 2\rho\tau_{xy} & 0 \\
-\rho\bigl(\tfrac23 G+\tau_{zz}\bigr) & 0 & 2\rho\tau_{xz} \\
0 & \rho\bigl(G+\tau_{xx}\bigr) & 0 \\
0 & 0 & \rho\bigl(G+\tau_{xx}\bigr) \\
-\rho\tau_{yz} & \rho\tau_{xz} & \rho\tau_{xy}
\end{bmatrix}.
\end{equation}
On the \((\alpha_1,\tau_{xx},\tau_{yy},\tau_{xy})\) rows these reduce to the two-dimensional \(N^u, N^v\) of Eq.~\eqref{eq:nc}.
Under in-plane data (\(w=0\), \(\partial_z=0\)) the shear components \(\tau_{xz}\) and \(\tau_{yz}\) remain zero, while \(\tau_{zz}\) develops but enters the tracked system only through the elastic energy.
The new column \(N^w\) carries the second shear pair \((w,\tau_{xz})\) through \(\rho(G+\tau_{xx})\) and the \(\tau_{zz}\) production through \(2\rho\tau_{xz}\), mirroring \(N^v\) under \(y\to z\).
The cross-stress \(\tau_{yz}\) is the only component coupling to all three velocity gradients, with coefficients \(-\rho\tau_{yz}, \rho\tau_{xz}, \rho\tau_{xy}\) in \(N^u, N^v, N^w\).

Two structural features stand out.
The two shear pairs \((v,\tau_{xy})\) and \((w,\tau_{xz})\) share the same tangential-velocity-gradient coefficient \(\rho(G+\tau_{xx})\), reflecting the doubly degenerate shear-wave family in three dimensions.
The cross-stress \(\tau_{yz}\) forms no independent wave family, so it introduces no new wave and varies only across the five waves already present.

\subsection{Characteristic structure}

The \(x\)-directed eigenvalues are
\begin{equation}
\label{eq:3d-eigs}
\lambda_{1\text{--}7}=u,
\qquad
\lambda_{s,\pm}^{(1)}=u\pm c_s,
\qquad
\lambda_{s,\pm}^{(2)}=u\pm c_s,
\qquad
\lambda_{f,\pm}=u\pm c_f,
\end{equation}
with \(c_s\) and \(c_f\) as in~\eqref{eq:wavespeeds}.
These are the eigenvalues of the three-dimensional reduced characteristic matrix, carrying the same pressure-row approximation as in Section~\ref{subsec:char}.
The shear eigenvalues \(u\pm c_s\) each have multiplicity two: one eigenspace carries \((v,\tau_{xy})\) and the other carries \((w,\tau_{xz})\).
The seven-fold degenerate contact family carries the material variables \((\alpha_1,m_1,m_2)\) and the transverse stresses \((\tau_{yy},\tau_{zz},\tau_{yz})\).
Across it, \(p\) and \(\tau_{xx}\) jump while the total pressure \(p_T = p - \tau_{xx}\) stays invariant.

\subsection{HLLD extension}

\subsubsection{Five-wave fan}
The five-wave fan \(S_L<S_{L*}^{\widehat D}<S_M<S_{R*}^{\widehat D}<S_R\) of Section~\ref{subsec:hlld} extends unchanged.
Both tangential pairs are resolved by the same inner wave speeds \[ S_{L*}^{\widehat D}=S_M-\frac{\widehat C^{\widehat D}}{\rho_L^{*}}, \qquad S_{R*}^{\widehat D}=S_M+\frac{\widehat C^{\widehat D}}{\rho_R^{*}}, \] with \(\widehat C^{\widehat D}\) given by Eq.~\eqref{eq:Chat}.
\subsubsection{Outer star states}
For \(K\in\{L,R\}\) the longitudinal and contact structure carries over unchanged from the two-dimensional construction.
The compression \(\xi_K\) and the star masses \(m_{j,K}^{*},\rho_K^{*}\) (Eqs.~\eqref{eq:xi}, \eqref{eq:hlld-partialmasses}), the volume fraction \(\alpha_{1,K}^{*,\widehat D}\) (Eq.~\eqref{eq:hlld-alpha}), the normal stress \(\tau_{xx,K}^{*,\widehat D}\) (Eq.~\eqref{eq:hlld-tauxx}) and transverse stress \(\tau_{yy,K}^{*,\widehat D}\) (Eq.~\eqref{eq:hlld-tauyy}), the contact values \(S_M,p_T^{*}\) (Eqs.~\eqref{eq:SM-hllc}--\eqref{eq:Tstar-hllc}), the normal momentum \((\rho u)_K^{*}=\rho_K^{*}S_M\) (Eq.~\eqref{eq:star-mom}), the pressure \(p_K^{*,\widehat D}=p_T^{*}+\tau_{xx,K}^{*,\widehat D}\) (Eq.~\eqref{eq:hlld-pstar}) and the energy \((\rho E)_K^{*}\) (Eq.~\eqref{eq:hlld-energy-star}) are unchanged, since the third direction leaves the normal acoustic and contact dynamics intact.
The first tangential pair is copied across the outer wave as in two dimensions (Eq.~\eqref{eq:hlld-outer-tangential}) and the second pair behaves identically,
\begin{equation}
\label{eq:3d-outer-tangential}
v_K^{*}=v_K,\quad \tau_{xy,K}^{*}=\tau_{xy,K},\qquad
w_K^{*}=w_K,\quad \tau_{xz,K}^{*}=\tau_{xz,K},\qquad
(\rho v)_K^{*}=\rho_K^{*}v_K,\quad (\rho w)_K^{*}=\rho_K^{*}w_K.
\end{equation}
The two transverse stresses with no two-dimensional counterpart follow the \(\tau_{yy}^{*,\widehat D}\) template.
The \(zz\) stress is its image under \(y\to z\),
\begin{equation}
\label{eq:3d-tauzz-star}
\tau_{zz,K}^{*,\widehat D}=\tau_{zz,K}
+\frac{\widehat\rho^{\widehat D}\bigl(\tfrac23\widehat G^{\widehat D}+\widehat\tau_{zz}^{\widehat D}\bigr)(S_M-u_K)}
       {\rho_K(S_K-u_K)},
\end{equation}
and the cross-stress carries the off-diagonal coefficient \(\widehat\tau_{yz}^{\widehat D}\),
\begin{equation}
\label{eq:3d-tauyz-star}
\tau_{yz,K}^{*,\widehat D}=\tau_{yz,K}
+\frac{\widehat\rho^{\widehat D}\,\widehat\tau_{yz}^{\widehat D}\,(S_M-u_K)}
       {\rho_K(S_K-u_K)}.
\end{equation}
These outer-star values are entered into the double-star increments~\eqref{eq:3d-tauzz-dblstar} and~\eqref{eq:3d-tauyz-dblstar} below.

\subsubsection{Inner star states}
Each pair has its own common double-star state, computed independently by the two-dimensional formulas~\eqref{eq:vstarstar}--\eqref{eq:tauxystarstar} applied separately to \((v,\tau_{xy})\) and \((w,\tau_{xz})\).
The transverse-stress increments mirror the two-dimensional formula, with each diagonal stress coupled to its own shear pair:
\begin{alignat}{2}
\label{eq:3d-tauyy-dblstar}
\tau_{yy,L}^{**,\widehat D}
&=
\tau_{yy,L}^{*,\widehat D}
+
\frac{2\widehat\rho^{\widehat D}\widehat\tau_{xy}^{\widehat D}}
     {\widehat C^{\widehat D}}
\left(v^{**,\widehat D}-v_L\right),
\\
\label{eq:3d-tauzz-dblstar}
\tau_{zz,L}^{**,\widehat D}
&=
\tau_{zz,L}^{*,\widehat D}
+
\frac{2\widehat\rho^{\widehat D}\widehat\tau_{xz}^{\widehat D}}
     {\widehat C^{\widehat D}}
\left(w^{**,\widehat D}-w_L\right).
\end{alignat}
The cross-stress \(\tau_{yz}\) receives increments from both shear pairs:
\begin{equation}
\label{eq:3d-tauyz-dblstar}
\tau_{yz,L}^{**,\widehat D}
=
\tau_{yz,L}^{*,\widehat D}
+
\frac{\widehat\rho^{\widehat D}}{\widehat C^{\widehat D}}
\Bigl(
\widehat\tau_{xz}^{\widehat D}
\bigl(v^{**,\widehat D}-v_L\bigr)
+
\widehat\tau_{xy}^{\widehat D}
\bigl(w^{**,\widehat D}-w_L\bigr)
\Bigr).
\end{equation}
The energy double-star sums the tangential work from both pairs:
\begin{equation}
\label{eq:3d-energy-dblstar}
(\rho E)_{L}^{**,\widehat D}
=
(\rho E)_{L}^{*}
+
\frac{\rho_L^{*}}{\widehat C^{\widehat D}}
\Bigl(
v^{**,\widehat D}\,\tau_{xy}^{**,\widehat D}
-v_L\,\tau_{xy,L}
+w^{**,\widehat D}\,\tau_{xz}^{**,\widehat D}
-w_L\,\tau_{xz,L}
\Bigr).
\end{equation}
The right-side (\(R\)) formulas use the same expressions with opposite signs on the increments.

\subsubsection{Anchored sampled flux}
Fan sampling and flux assembly follow the two-dimensional construction, with the anchored augmented flux \(F^{\widehat D}(W) = F(W) - N^u(\widehat W^{\widehat D})\,u(W) - N^v(\widehat W^{\widehat D})\,v(W) - N^w(\widehat W^{\widehat D})\,w(W)\), the three-dimensional analog of Eq.~\eqref{eq:anchored-flux} with the additional \(N^w\,w\) term.

\subsection{Directional permutations}

The \(y\)- and \(z\)-directed subsystems are obtained by the index permutations: \[
\renewcommand{\arraystretch}{1.3}
\begin{array}{l@{\quad}l}
y\text{-sweep:} &
(u,v,w)\to(v,u,w),
\quad
(\tau_{xx},\tau_{xy},\tau_{xz},\tau_{yy},\tau_{zz},\tau_{yz})
\to
(\tau_{yy},\tau_{xy},\tau_{yz},\tau_{xx},\tau_{zz},\tau_{xz}),
\\
z\text{-sweep:} &
(u,v,w)\to(w,u,v),
\quad
(\tau_{xx},\tau_{xy},\tau_{xz},\tau_{yy},\tau_{zz},\tau_{yz})
\to
(\tau_{zz},\tau_{xz},\tau_{yz},\tau_{xx},\tau_{yy},\tau_{xy}).
\end{array}
\] In each case the normal velocity and normal stress of the \(x\)-sweep are replaced by their directional counterparts, and the tangential pairs and cross-stress are relabeled accordingly.

\clearpage
\section{Complex-inclusion robustness demonstration}\label{app:complex}

The shock--cylinder case of Section~\ref{subsec:s62-shock-cylinder} compares the solvers on a single smooth inclusion.
This appendix adds a qualitative robustness demonstration on a deliberately challenging geometry.
The same Mach $1.006$ water shock, with the identical post-shock state, impacts four embedded epoxy inclusions: a rotated rectangle, an ellipse, a rotated triangle, and a circle, with overlapping regions merged by union.
The geometry combines curved and non-grid-aligned interfaces, sharp corners, fine features, and narrow gaps that set up many internal reflections.
The domain matches the transient-cylinder case of Section~\ref{subsec:s61-acoustic-scattering}, the grid is $1200\times800$, and the case is run to $t = 1500\,\mu$s.
Anchored HLLD completes the run, with the deviatoric stresses confined to the inclusions and the interfaces clean throughout (Figure~\ref{fig:complex-shock-contour}).

\begin{figure}[htbp]
\centering
\includegraphics[width=\textwidth]{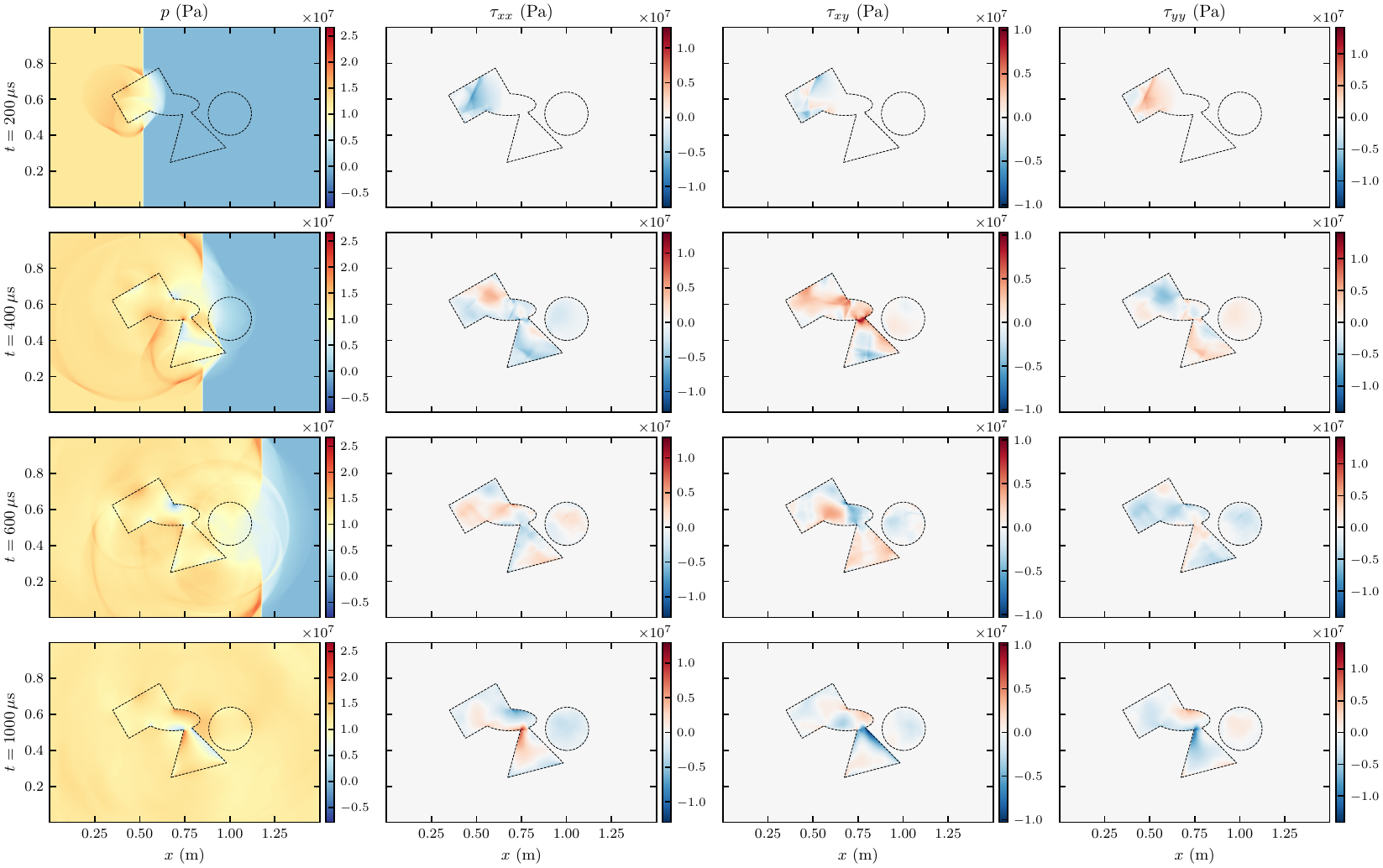}
\caption{Complex-inclusion shock case: $p$, $\tau_{xx}$, $\tau_{xy}$, $\tau_{yy}$ at $t = 200$, $400$, $600$, $1000\,\mu$s (anchored HLLD, $1200\times800$). Dashed contours mark the instantaneous material interfaces ($\alpha_2 = 0.5$).}
\label{fig:complex-shock-contour}
\end{figure}

\clearpage
\section*{CRediT authorship contribution statement}

\noindent\textbf{Yanjun Zhang:} Conceptualization, Methodology, Software, Validation, Formal analysis, Investigation, Data curation, Visualization, Writing -- original draft, Writing -- review \& editing.

\noindent\textbf{Tim Colonius:} Conceptualization, Methodology, Resources, Supervision, Funding acquisition, Writing -- review \& editing.

\section*{Acknowledgments}

The authors thank Spencer H.~Bryngelson for his helpful comments on the manuscript and his assistance with reviewing, refining, and integrating the solver implementation into MFC.
They also gratefully acknowledge his continued development and maintenance of the MFC codebase.

This work was supported by the National Science Foundation under Grant No.~2529317.

Computations were performed on Bridges-2 at the Pittsburgh Supercomputing Center~\cite{BrownEtAl2021Bridges2} through allocation MCH260104 from the Advanced Cyberinfrastructure Coordination Ecosystem: Services \& Support (ACCESS) program~\cite{BoernerEtAl2023ACCESS}.
The ACCESS program is supported by National Science Foundation grants 2138259, 2138286, 2138307, 2137603, and 2138296.

Computations were also performed on Frontier under OLCF project CFD154.
An award of computer time was provided by the INCITE program.
This research also used resources of the Oak Ridge Leadership Computing Facility, which is a DOE Office of Science User Facility supported under Contract DE-AC05-00OR22725.

\section*{Declaration of competing interests}

The authors declare that they have no known competing financial interests or personal relationships that could have appeared to influence the work reported in this paper.

\section*{Code availability}

The anchored HLLD Riemann solver presented in this work is implemented in the open-source MFC codebase, available at \url{https://github.com/MFlowCode/MFC}.

\section*{Declaration of generative AI and AI-assisted technologies in the manuscript preparation process}

During the preparation of this work, the authors used ChatGPT and Codex (OpenAI), and Claude Code (Anthropic) to assist with improving the manuscript's language, clarity, and organization, identifying relevant literature, and reviewing the manuscript for potential errors and inconsistencies.
The authors subsequently reviewed and edited the AI-assisted content, independently verified the relevant references, and take full responsibility for the content of the published article.

\end{document}